\documentclass[11pt,reqno]{amsart}

\usepackage{tikz}
\usetikzlibrary{arrows,positioning, calc}
\tikzstyle{vertex}=[draw,fill=black!15,circle,minimum size=20pt,inner sep=0pt]

\title[Parametric set-theoretic YBE: $p$-racks, solutions $\&$ quantum algebras]{Parametric set-theoretic Yang-Baxter equation: $p$-racks, solutions $\&$ quantum algebras}
\author[Anastasia Doikou]{Anastasia Doikou}

\address[Anastasia Doikou] {Department of Mathematics, Heriot-Watt University,
Edinburgh EH14 4AS $\&$ Maxwell Institute for Mathematical Sciences, Edinburgh EH8 9BT, UK}
\email{a.doikou@hw.ac.uk}

  \usepackage{latexsym} 
  \usepackage[all]{xy}
  \usepackage{amsfonts} 
  \usepackage{amsthm} 
  \usepackage{amsmath} 
  \usepackage{amssymb}
  \usepackage{pifont}  
  \usepackage{enumerate}
  \usepackage{dcolumn}
  \usepackage{comment}
  \usepackage{hyperref}  
  \usepackage{tikz-cd}
\usepackage[english]{babel}
  \usepackage{amsfonts} 
\usepackage[utf8]{inputenc}
\usepackage{graphicx}
\usepackage{subcaption}
\usepackage[colorinlistoftodos]{todonotes}
\usepackage{indentfirst}
\usepackage{pifont}  
  \usepackage{enumerate}
  \usepackage{dcolumn}
  \usepackage{comment}
\usepackage[capitalise]{cleveref}

 \usepackage{tikz}
\usetikzlibrary{arrows}

 \newcolumntype{2}{D{.}{}{2.0}}
  \xyoption{2cell}

\newcommand{\hiddenpower}[2] { \ifnum \numexpr#2=1 #1 \else #1^#2 \fi }
\numberwithin{equation}{section}

\def\be{\begin{equation}}
\def\ee{\end{equation}}
\def\ba{\begin{eqnarray}}
\def\ea{\end{eqnarray}}

\newcommand{\cal}{\mathcal}

\usepackage{xparse}
\ExplSyntaxOn
\newcounter{diff_order}
\newcounter{diff_power}

\newcommand{\rawdiff}[3]
{
	\setcounter{diff_order}{0}
	\clist_map_inline:nn{#3}{\stepcounter{diff_order}}
	
	\frac{\hiddenpower{#1}{\thediff_order} #2}
	{
		\def\old_var{DefaultValue}
		\setcounter{diff_power}{0}
		
		\clist_map_inline:nn{#3}
		{
			\def\new_var{##1}
			\ifnum \thediff_power=0
				\stepcounter{diff_power}
			\else
				\tl_if_eq:NNTF \new_var \old_var
				{\stepcounter{diff_power}}
				{
					#1 \hiddenpower{\old_var}{\thediff_power}
					\setcounter{diff_power}{1}
				}
			\fi

			\def\old_var{##1}
		}
		
		#1 \hiddenpower{\old_var}{\thediff_power}
	}
}
\DeclareMathOperator{\EEnd}{End}

\def\Label#1{\label{#1}\ifmmode\llap{[#1] }\else 
  \marginpar{\smash{\hbox{\tiny [#1]}}}\fi} 
  \def\Label{\label} 

\ExplSyntaxOff

\newlength{\bibitemsep}
\newlength{\bibparskip}
\let\oldthebibliography\thebibliography
\renewcommand\thebibliography[1]{%
  \oldthebibliography{#1}%
  \setlength{\parskip}{\bibitemsep}%
  \setlength{\itemsep}{\bibparskip}%
}

\newtheorem{thm}{Theorem}[section]
\newtheorem{lemma}[thm]{Lemma}
\newtheorem{cor}[thm]{Corollary}
\newtheorem{pro}[thm]{Proposition}
\newtheorem{defn}[thm]{Definition}

\newtheorem{rem}[thm]{Remark}
\newtheorem{exa}[thm]{Example}

\newcommand{\id}{\operatorname{id}}

\newenvironment{wideregion}[1][0mm]
    {\vspace{#1}\list{}{\leftmargin-2cm\rightmargin\leftmargin}\item\relax}
    {\endlist}

\newenvironment{widegather }{\wideregion[-9mm]\gather}{\endgather\endwideregion}

\begin{document}
\vskip 0.8in

\hfill
\begin{abstract} 
The theory of the parametric set-theoretic Yang-Baxter equation is established from a purely algebraic point of view. 
The first step towards this objective is the introduction  of certain generalizations of the familiar shelves and racks called parametric ($p$)-shelves and racks. These objects satisfy a {\it parametric self-distributivity} condition and lead to solutions of the Yang-Baxter equation. Novel, non-reversible solutions are obtained from $p$-shelf/rack solutions by a suitable parametric twist, whereas all reversible set-theoretic solutions are reduced to the identity map via a parametric twist. The universal algebras associated to both $p$-rack and generic parametric, set-theoretic solutions are next presented and the corresponding universal ${\cal R}$-matrices are derived. The admissible universal Drinfel'd twist is constructed allowing the derivation of the general set-theoretic universal ${\cal R}$-matrix.
By introducing the concept of a parametric coproduct we prove the existence of a parametric co-associativity. We show that the parametric coproduct is an algebra homomorphism and the universal ${\cal R}$-matrices satisfy intertwining relations with the algebra coproducts.\\
\\
\end{abstract}

\maketitle
 
\date{}



\section{Introduction}

\noindent 
The Yang-Baxter equation (YBE) is a central object in contemporary mathematics and mathematical physics. The breadth of its applications extends from one dimensional
statistical systems and integrable quantum field theories to quantum group theory \cite{Drinfel'd, Jimbo1, Jimbo2} and low-dimensional topology \cite{Jo82, Kauff}
and a plethora of other areas in mathematics and physics.  The equation was first introduced in a purely physical context in \cite{Yang} as the main mathematical tool for the investigation 
of quantum systems with many particle interactions, and in \cite{Baxter} 
for the study of statistical model known as the anisotropic Heisenberg magnet. The idea of
set-theoretic solutions to the Yang-Baxter equation was suggested later by Drinfel'd \cite{Dr92} and since then, set-theoretic solutions
have been extensively investigated primarily by means of representations of the braid group, but almost exclusively for the parameter free Yang-Baxter equation (see for instance \cite{EtScSo99, GuaVen, Ru05, Ru07}).
The investigation of set-theoretic solutions of the  non-parametric Yang-Baxter equation and the associated algebraic structures is a highly active research field that has been particularly prolific, given that a significant number of related studies has been produced over the past
few years (see for instance \cite{Bachi}--\cite{LebMan}  \cite{DoiSmo}--\cite{DoRy23}, \cite{JesKub, Smo1, Smo2}, \cite{GatMaj}--\cite{Gat2}, \cite{Pili, JesKub}).   
The theory of the set-theoretic Yang-Baxter equation has numerous significant connections to distinct mathematical
areas, such as group theory,  algebraic number theory,  Hopf-Galois extensions, non-commutative rings,  knot theory, Hopf algebras and quantum groups,  universal algebras, groupoids,  trusses and heaps \cite{Brz:Lietruss}, pointed Hopf algebras, {Yetter-Drinfel'd} modules and Nichols algebras (see for instance among \cite{Andru, braceh, Bachi, Brz:Lietruss, EtScSo99}, 
\cite{Pili}-\cite{Jesp2},\cite{Jo82}-\cite{LebVen}, \cite{Ru07,Smo1}).
 Moreover,  interesting links  between the
set-theoretic Yang-Baxter equation and geometric crystals \cite{Crystals, Crystal2}, or soliton cellular automatons \cite{Cell, Cell2} have been shown.  Concrete connections with quantum spin-chain like systems were also made in \cite{DoiSmo, DoiSmo2}.

The main objectives of the present investigation are the derivation of set-theoretic solutions of the parametric Yang-Baxter equation as well as the rigorous formulation of the associated 
Yang-Baxter algebraic structures underpinning the parametric set-theoretic Yang-Baxter equation, 
and the identification of the universal ${\cal R}$-matrices.
Set-theoretic solutions for the parametric Yang-Baxter equation (Yang-Baxter 
maps) have been studied up to date only
in the context of classical discrete integrable systems connected also to the notion of
Darboux-B\"acklund transformation in the Lax pair formulation 
\cite{Adler, Papa,  Papa2, Veselov} and the refactorization frame.  In this investigation for the first time to our knowledge an entirely algebraic analysis for the parametric set-theoretic Yang-Baxter equation is undertaken and purely algebraic solutions are produced.
Earlier works on the algebraic structures of the non-parametric set-theoretic Yang-Baxter equation provide a basic algebraic blueprint \cite {Doikou1, DoGhVl, DoRySt}, however the parametric case turns out to be considerably more involved in comparison as dealing with the various parameters involved is a combinatorial problem on its own. This statement concerns all the main findings of the present investigation, such as  the proof of the main Theorem \ref{le:lndsol} as well as the formulation of the underlying quantum algebras associated to the set-theoretic parametric version of the Yang-Baxter equation. 
For instance the usual notion of (co)-associativity does not apply any more, however a version of parametric (co)-associativity together with consistent recursion relations are provided making the notion of a parametric co-product a well defined mathematical object as an algebra homomorphism.

We introduce at this point the parametric set-theoretic Yang-Baxter equation.
Let $X, \ Y \subseteq X$ be non-empty sets,  $z_{i},z_{j} \in Y,$ $i,j \in \mathbb {\mathbb Z}^+$ and $R^{z_{ij}}:X\times X\rightarrow X\times X,$ such that for all $x,y\in X,$
 $~R^{z_{ij}}(y, x)= \big (\sigma^{z_{ij}} _{x}(y), \tau^{z_{ij}}_{y}(x)\big ).$
The notation $z_{ij}$ denotes dependence 
on $(z_i,  z_j).$  We say that $(X, R^{z_{ij}})$ is a solution of the parametric, set-theoretic Yang-Baxter equation (or simply a solution) if 
\begin{equation}
R_{12}^{z_{12}} \ R^{z_{13}}_{13}\ R_{23}^{z_{23}} =  R_{23}^{z_{23}}\ R^{z_{13}}_{13}\  R_{12}^{z_{12}}, \label{YBE}
\end{equation}
where
$R_{12}^{z_{ij}}(c,b,a) = (\sigma^{z_{ij}}_b(c) ,\tau^{z_{ij}}_c(b), a),$ 
$~R_{13}^{z_{ij}}(c,b,a) = (\sigma^{z_{ij}}_a(c) ,b, \tau^{z_{ij}}_c(a))$ and $~R_{23}^{z_{ij}}(c,b,a) = (c, \sigma^{z_{ij}}_a(b) ,\tau^{z_{ij}}_b(a)).$ We say that $R^{z_{ij}}$ is a left non-degenerate if for all $x\in X,$ $z_{i},z_{j}\in Y,$ $\sigma^{z_{ij}}_{x}$ is a bijective function 
and  non-degenerate if both $\sigma^{z_{ij}}_{x},\ \tau^{z_{ij}}_{x}$ are bijective functions. Also, the solution $(X, R^{z_{ij}})$  is called ``reversible'' \cite{Adler, Papa, Papa2} if $R^{z_{21}}_{21} R^{z_{12}}_{12} = \id$. Interestingly, all the solutions from the point of view of discrete integrable systems \cite{Adler, Papa, Papa2, Veselov} or the re-factorization approach are reversible and it is shown  in Proposition \ref{Th.r->r'}, Remark \ref{rem1} that they are reduced to the identity map.  Here, for the first time we present non-reversible solutions of the parametric set-theoretic Yang-Baxter equation.

We summarize below the main outcomes of each of the following sections. In Section 2, we introduce the notions of parametric ($p$)-shelves and racks to describe solution of the parametric set-theoretic Yang-Baxter equation. These algebraic objects 
may be thought of as parametric generalizations of the familiar shelves and racks, they satisfy a generalized {\it parametric self-distributivity} and naturally yield solutions of the Yang-Baxter equation. Parametric shelf or rack solutions of the Yang-Baxter equation are derived here for the first time. We then show in the second subsection that every left non-degenerate solution is Drinfel'd equivalent to a solution given by a shelf, (see \cref{Th.r->r'}. In fact, by \cref{le:lndsol}). Every left non-degenerate solution can be obtained from a shelf solution by finding a special map of the shelf itself, which we call \emph{twist}, \cref{def:twist:shelf}. Bijective solutions naturally correspond to $p$-racks.
In the third subsection we introduce the notions of parametric Yang-Baxter operators and structures by generalizing the idea of braiding to the parametric case. This section serves mainly as a warm up to Section 3, given that the generalized Yang-Baxter structures encode part of the associated underlying universal algebras, which are introduced next.

In Section 3,  we focus on the linearized version of the set-theoretic Yang-Baxter equation. 
In the first subsection we introduce the $p$-rack algebra and we then construct the associated universal 
${\cal R}$-matrix. 
By introducing the concept of a parametric coproduct we prove the existence of a parametric co-associativity. These notions allow us to show that the parametric coproduct is an algebra homomorphism and the universal ${\cal R}$-matrices intertwine with the algebra generators coproducts.
Moreover, due to the parametric coassociativity the $n$-coproducts are consistently expressed in terms of $2^{n-2}$ suitable parametric binary trees.
This algebraic structure is a rather new paradigm of quantum algebra, as the parameters are part of the universal structure. 
We also note that a quantum integrability statement is presented, as a family of mutually commuting non-local objects is identified. In the second subsection, we suitably extend the $p$-rack algebra and present the decorated $p$-rack and the $p$-set Yang-Baxter algebra. By means of a suitable admissible universal Drinfel'd twist 
we construct the universal-set theoretic ${\cal R}-$matrix. Fundamental representations of the aforementioned  
algebras are also considered leading naturally to the $p$-rack and general set-theoretic solutions of the parametric Yang-Baxter equation.

\section{Solving the parametric set-theoretic YBE} 

\noindent In this section we formulate the basic problem, which is the derivation of solutions of the parametric set-theoretic Yang-Baxter equation. Specifically, we identify the sets of conditions that such solutions satisfy and then we move one to the identification of concrete solutions. The key step into achieving these goals in the presentation of some new algebraic objects called $p$-shelves and $p$-racks that satisfy parametric self-distributivity conditions. Use of the newly introduced $p$-shelves/racks together with the identification of a suitable parametric twist allows the derivation of generic parametric set-theoretic solutions.

\subsection{Parametric shelves and racks}
We are first exploring a certain class of solutions of the parametric set-theoretic equation that are generalizations of the shelf/rack \cite{Shelf-history2} type solutions of the non-parametric YBE. Such solutions are derived here for the first time.
We start our analysis in this section with a definition that generalizes the notion of shelves and racks by introducing the {\it parametric shelves and racks} ($p$-shelves and $p$-racks).  Note that henceforth we write  $z_{i,j,k, \ldots} \in Y$ as a shorthand notation for $z_i, z_j, z_k,  \ldots \in Y,$ $i,j,k, \ldots \in {\mathbb Z}^+.$

\begin{defn} ($p$-shelves and $p$-racks) Let $X, \ Y \subseteq X$ be  non-empty sets. We define for all $z_{i,j} \in Y,$ the binary operation $\triangleright_{z_{ij}}: X \times X \to X,$ $(a,b) \mapsto a \triangleright_{z_{ij}} b.$ The pair $\left(X,\,\triangleright_{z_{ij}} \right)$ 
is said to be a \emph{left parametric ($p$)-shelf} if\, $\triangleright_{z_{ij}}$\, satisfies the generalized left $p$-self-distributivity:
    \begin{equation}\label{shelf}
a\triangleright_{z_{ik}}\left(b\triangleright_{z_{jk}} c\right)
        = \left(a\triangleright_{z_{ij}} b\right)\triangleright_{z_{jk}}\left(a\triangleright_{z_{ik}} c\right) 
    \end{equation}
    for all $a,b,c\in X,$ $z_{i,j,k} \in Y.$ Moreover, a left $p$-shelf $\left(X,\,\triangleright_{z_{ij}} \right)$ is called 
     a left $p$-rack if the maps $L^{z_{ij}}_a:X\to X$ defined by $L^{z_{ij}}_a\left(b\right):= a\triangleright_{z_{ij}} b$, for all $a,b, \in X,$ $z_{i,j} \in Y,$ are bijective.
\end{defn}
From now on whenever we say $p$-shelf or $p$-rack we mean left $p$-shelf or left $p$-rack.

\begin{pro} \label{prose}
Let $X, \ Y \subseteq X$ be non-empty sets. We define for $z_{i, j} \in Y$ 
the binary operation $\triangleright_{z_{ij}}:  X \times X \to X,$ $(a,b) \mapsto a \triangleright_{z_{ij}} b.$ Then $R^{z_{ij}}:X \times X \to X \times X$, such that for all $a,b \in X, $ $z_{i,j} \in Y,$ $R^{z_{ij}}(b,a) = (b, b \triangleright_{z_{ij}} a)$ is a solution of the parametric set-theoretic Yang-Baxter equation if and only if $(X, \triangleright_{z_{ij}})$ is a $p$-shelf.
\end{pro}
\begin{proof}
Let $R^{z_{ij}}$ be a solution of the parametric Yang-Baxter equation, then for all $a,b,c \in X$ and $z_{1, 2,3} \in Y$ the LHS of the Yang-Baxter equation (\ref{YBE}) gives
\begin{equation}
    R_{12}^{z_{12}} \ R^{z_{13}}_{13}\ R_{23}^{z_{23}}(a,b,c)  = (a, a\triangleright_{z_{12}}b,  a\triangleright_{z_{13}}( b\triangleright_{z_{23}}c)) \label{lhs} \end{equation}
    whereas the RHS gives
    \begin{equation}    
    R_{23}^{z_{23}}\ R^{z_{13}}_{13}\  R_{12}^{z_{12}}(a,b,c)  =  (a, a\triangleright_{z_{12}}b,  \left(a\triangleright_{z_{12}} b\right)\triangleright_{z_{23}}\left(a\triangleright_{z_{13}} c\right) )  \label{rhs} \end{equation}
    Equating (\ref{lhs}) and (\ref{rhs}) we conclude:
\[a\triangleright_{z_{13}}\left(b\triangleright_{z_{23}} c\right)
        = \left(a\triangleright_{z_{12}} b\right)\triangleright_{z_{23}}\left(a\triangleright_{z_{13}} c\right),\]
    i.e. $(X, \triangleright_{z_{ij}})$ is a $p$-shelf.  Conversely, if $(X, \triangleright_{z_{ij}})$ is a $p$-shelf then automatically the map $R^{z_{ij}}$ is a solution of the parametric set-theoretic Yang-Baxter equation.
\end{proof}

\begin{lemma}
    Let $(X, \triangleright_{z_{ij}})$ be a $p$-rack and $R^{z_{ij}}: X\times X \to X \times X,$ 
    $R^{z_{ij}} (a ,b ) = (a, a\triangleright_{z_{ij}} b)$ be a non-degenerate solution of the parametric set-theoretic YBE, i.e. $a\triangleright_{z_{ij}}$ is a bijection for all $a\in X,$ $z_{i,j} \in Y$. Then, $R^{z_{ij}}: X \times X \to X \times X$ is invertible with $(R^{z_{ij}})^{-1} (a, a\triangleright_{z_{ij}} b) = (a,b). $
\end{lemma}
\begin{proof} The proof is straightforward, using the fact that $a  \triangleright_{z_{ij}}: X \to X$ is a bijection.
\end{proof}

Before we further proceed it is useful to recall the notion of skew braces as 
this will allow us to derive concrete solutions of the parametric set-theoretic Yang-Baxter equation.
\begin{defn}\cite{Ru05}-\cite{Ru19}, \cite{GuaVen}.
A {\it left skew brace} is a set $B$ together with two group operations $+,\circ :B\times B\to B$, 
the first is called addition and the second is called multiplication, such that for all $ a,b,c\in B$,
\begin{equation}\label{def:dis}
a\circ (b+c)=a\circ b-a+a\circ c.
\end{equation}
If $+$ is an abelian group operation $B$ is called a 
{\it left brace}.
Moreover, if $B$ is a left skew brace and for all $ a,b,c\in B$ $(b+c)\circ a=b\circ a-a+c \circ a$, then $B$ is called a 
{\it two sided skew brace.}  Analogously if $+$ is abelian and $B$ is a skew brace,  then $B$ is called a {\it two sided brace}.
\end{defn}
The additive identity of a skew brace $B$ will be denoted by $0$ and the multiplicative identity by $1$.  
In every skew brace $0=1$. 

From now on when we say skew brace we mean left skew brace.
We may now state the following proposition regarding concrete solutions of the parametric set-theoretic
Yang-Baxter equation coming from skew braces.

Before we proceed with the construction of $p$-racks we give a simple example of a brace.
\begin{exa}[See \cite{BrzRyb:con} Corollary 3.14]\label{ex:cyclicbraces}
Let $\mathrm{U}(\mathbb{Z}/2^m\mathbb{Z})$ denote a set of invertible integers modulo $2^m$, for some $m \in \mathbb{N}$. 
Then a triple $(\mathrm{U}(\mathbb{Z}/2^m\mathbb{Z}),+_1,\circ )$ is a brace, where $a+_1b=a-1+b,$ 
for all $a,b\in \mathrm{U}(\mathbb{Z}/2^m\mathbb{Z})$, $+$ and $\circ $ are addition and multiplication of integer numbers modulo 
$2^m$, respectively. 
For instance, 1. for $m=1,$ $\mathrm{U}(\mathbb{Z}/2\mathbb{Z}) =\{1\},$ 2. for $m=2,$ $\mathrm{U}(\mathbb{Z}/2^2\mathbb{Z}) = \{1,\ 3\},$ 3. for $m=3,$ $\mathrm{U}(\mathbb{Z}/2^m\mathbb{Z})
= \{1,\ 3,\ 5,\ 7\},$ etc.\end{exa}

\begin{pro} \label{prora1}
Let $(X,+, \circ)$ be a skew brace and $Y\subseteq X,$ such that  
 \begin{itemize}
   \item for all $a,b \in X,$ $z\in Y,$ $(a+b)\circ z = a \circ z -z +b \circ z,$
   \item $z \in Y$ are central in $(X, +).$
    \end{itemize}
     Define also for all $z_{i,j}\in Y$ and some $\xi \in X$ the binary operation $\triangleright_{z_{ij}}: X \times X \to X,$ such that for all $a,b \in X,$
$a \triangleright_{z_{ij}} b = -\xi\circ a\circ z_i \circ z_j^{-1} + \xi\circ b + a \circ z_i \circ z_j^{-1}.$ 
Then the map $R^{z_{ij}}: X\times X \to X \times X,$ 
such that for all $a,b \in X,$ $z_{i,j} \in Y,$ 
\[ R^{z_{ij}}(a,b) = (a, a \triangleright_{z_{ij}} b) \]
is a solution of the parametric Yang-Baxter equation.  Moreover, the map $R^{z_{ij}}$ is invertible.
\end{pro}
\begin{proof}
It suffices to show that the binary operation satisfies the $p$-self-distributivity condition, 
i.e.  $(X, \triangleright_{z_{ij}})$ is a $p$-shelf. Let $z_{i, j, k} \in Y,$ then the LHS of condition (\ref{shelf})

\begin{eqnarray}
a\triangleright_{z_{ik}}\left(b\triangleright_{z_{jk}} c\right)   
& =&  -\xi\circ a \circ z_i \circ z_k^{-1} +\xi \circ (b\triangleright_{z_{jk}} c) + a \circ z_i \circ z_k^{-1}  \nonumber\\
  & =& -\xi\circ a \circ z_i \circ z_k^{-1} +\xi -\xi\circ \xi\circ  b \circ z_j \circ z_k^{-1} +\xi\circ\xi \circ c 
  \nonumber\\ & &-\xi + \xi\circ b \circ z_j \circ z_k^{-1}  +a \circ z_i \circ z_k^{-1}. \nonumber
  \end{eqnarray}
Similarly the RHS of condition (\ref{shelf}):
  \begin{eqnarray}
 \left(a\triangleright_{z_{ij}} b\right)\triangleright_{z_{jk}}\left(a\triangleright_{z_{ik}} c\right) 
& =&  -\xi\circ \left(a\triangleright_{z_{ij}} b\right) \circ z_j \circ z_k^{-1}  + \xi\circ (a\triangleright_{z_{ik}} c) + \left(a\triangleright_{z_{ij}} b\right) \circ z_j \circ z_k^{-1} \nonumber\\
&=& -\xi \circ(-\xi \circ a \circ z_i\circ z_j^{-1} + \xi\circ b +a \circ z_i \circ z_j^{-1})\circ z_j \circ z_k^{-1}\nonumber\\
& &+\xi \circ (-\xi\circ a\circ z_i \circ z_k^{-1} + \xi \circ c + a\circ z_i \circ z_k^{-1})
\nonumber\\
& & (-\xi \circ a \circ z_i \circ z_j^{-1} + \xi\circ b + a \circ z_i \circ z_j^{-1}) \circ z_j \circ z_k^{-1}
\nonumber\\
  & =& -\xi\circ a \circ z_i \circ z_k^{-1} +\xi -\xi\circ \xi\circ  b \circ z_j \circ z_k^{-1} +\xi\circ\xi \circ c 
  \nonumber\\ & &-\xi + \xi\circ b \circ z_j \circ z_k^{-1}  +a \circ z_i \circ z_k^{-1}.\nonumber
  \end{eqnarray}  
Indeed, LHS $=$ RHS which concludes our proof.

  Moreover, there exists 
  $a \triangleright_{z_{ij}}^{-1}: X \to X,$ such that  $a \triangleright_{z_{ij}}^{-1}(a \triangleright_{z_{ij}}b) =a \triangleright_{z_{ij}}(a \triangleright_{z_{ij}}^{-1}b)=b$ and we immediately extract from $a \triangleright_{z_{ij}}(a \triangleright_{z_{ij}}^{-1}b)=b$ that $a \triangleright^{-1}_{z_{ij}}b = 
  a\circ z_i \circ z_j^{-1} -\xi^{-1} +\xi^{-1} \circ b - \xi^{-1} \circ a \circ z_i \circ z_j^{-1} + \xi^{-1},$  i.e.  $a \triangleright_{z_{ij}}$ is a bijection for all $a\in X,$ $z_{i,j} \in Y.$  Hence, $(X, \triangleright_{z_{ij}})$ is a $p$-rack, and  $R^{z_{ij}}$ is invertible. 
  \end{proof}

\begin{rem} \label{rem0} 
If $(X,  +, \circ)$ is a brace,  i.e.  $(X, +)$ is an abelian group, and $\xi=1,$ then for all $a,b \in X,$ $z_{i,j} \in Y,$ 
$a \triangleright_{z_{ij}} b =  -a\circ z_i \circ z_j^{-1}+b + a \circ \ z_i \circ {z}_j^{-1} = b,$ and hence $R^{z_{ij}} = \id.$ 
\end{rem}

\begin{exa} \label{good1} 
Recall example \ref{ex:cyclicbraces} ($X  = \mathrm{U}(\mathbb{Z}/2^m\mathbb{Z})$), fix some $\xi \in X$ and define for all $a,b\in X,$ $z_{i}, z_{j} \in X$ $a\triangleright_{z_{ij}} b = -_1\xi\circ   a \circ z_i \circ z_j^{-1} +_1 \xi \circ b +_1 a\circ z_i \circ z_j^{-1}.$ Recall that for all $a,b \in X$ $a+_1 b = a -1 +b,$ then for all $a\in X,$ $-_1 a = 1 -a +1,$ and hence $a\triangleright_{z_{ij}} b = -\xi\circ   a \circ z_i \circ z_j^{-1} + \xi \circ b + a\circ z_i \circ z_j^{-1},$ where $+,$ and $\circ$ are the addition and multiplication in integers $mod ~m.$ For instance consider $\mathrm{U}(\mathbb{Z}/2^3\mathbb{Z}) = \{1,3,5,7\} $ and $\xi =3,$ then all $a\triangleright_{z_{ij}}b$ can be directly computed for all $a,b,z_i,z_j \in X.$ Specifically, we see that $a\triangleright_{z_{ij}} b \neq a\triangleright_{z_{i'j'}}b$ if $(z_i, z_j)\neq (z_{i'}, z_{j'}).$ Indeed, we choose $\xi =3$ and compute for example, $1\triangleright_{13}3 = 3,$ $\ 1\triangleright_{15}3 =7.$

\end{exa}

We now focus on the  more general solution of the set-theoretic Yang-Baxter equation of the type 
$R^{z_{ij}}: X \times X \to X \times X,$ such that for all 
$a, b \in X,$  $z_{i,j} \in Y,$
\[ R^{z_{ij}} (b,a) = (\sigma^{z_{ij}}_a(b), \tau_b^{z_{ij}}(a)).\]
We next derive the conditions satisfied by the general set-theoretic solution of the parametric Yang-Baxter equation. 
\begin{pro} 
Let $X,\ Y\subseteq X,$ be non-empty sets, and define for all $a,b\in X,$ $z_{i,j} \in Y,$  the maps $\sigma_a^{z_{ij}}, \tau_b^{z_{ij}}: X \to X,$ 
$b \mapsto \sigma_a^{z_{ij}}(b)$ and $a \mapsto \tau_b^{z_{ij}}(a).$ 
Then $R^{z_{ij}}:X \times X \to X \times X$, 
such that  for all $a,b\in X,$ $z_{i,j}\in Y,$ $R^{z_{ij}}(b,a) = (\sigma_a^{z_{ij}}(b), \tau_b^{z_{ij}}(a))$ is a solution of the parametric set-theoretic Yang-Baxter equation if and only if, for all $z_{1,2,3} \in Y,$
\begin{eqnarray}
&&  \sigma^{z_{13}}_a(\sigma^{z_{12}}_b(c)) = \sigma^{z_{12}}_{\sigma^{z_{23}}_a\left(b\right)}(\sigma^{z_{13}}_{\tau^{z_{23}}_b\left(a\right)}(c)) \label{1}\\ 
&& \tau^{z_{13}}_c(\tau^{z_{23}}_b(a)) =\tau^{z_{23}}_{\tau^{z_{12}}_c\left(b\right)}(\tau^{z_{13}}_{\sigma^{z_{12}}_b\left(c\right)}(a)) \label{2}\\
&& \sigma^{z_{23}}_{\tau^{z_{13}}_{\sigma^{z_{12}}_b\left(c\right)}\left(a\right)}
    (\tau^{z_{12}}_c\left(b\right)) 
   = \tau^{z_{12}}_{\sigma^{z_{13}}_{\tau^{z_{23}}_b\left(a\right)}\left(c\right)}(\sigma^{z_{23}}_a\left(b\right) ). \label{3} \end{eqnarray}
\end{pro}
\begin{proof} Let $R^{z_{ij}}$ be a solution.
We compute explicitly the LHS and RHS of the parametric Yang-Baxter equation (\ref{YBE}). The LHS of the Yang-Baxter equation gives, $a,b,c \in X,$ $z_{1,2,3} \in Y,$
    \begin{equation}
    R_{12}^{z_{12}} \ R^{z_{13}}_{13}\ R_{23}^{z_{23}}(c,b,a)  = (\sigma^{z_{12}}_{\sigma^{z_{23}}_a\left(b\right)}(\sigma^{z_{13}}_{\tau^{z_{23}}_b\left(a\right)}(c)), \tau^{z_{12}}_{\sigma^{z_{13}}_{\tau^{z_{23}}_b\left(a\right)}\left(c\right)}(\sigma^{z_{23}}_a\left(b\right) ),  \tau^{z_{13}}_c(\tau^{z_{23}}_b(a))), \label{lhs2} \end{equation}
    whereas the RHS gives
    \begin{equation}    
    R_{23}^{z_{23}}\ R^{z_{13}}_{13}\  R_{12}^{z_{12}}(c,b,a)  =  (\sigma^{z_{13}}_a(\sigma^{z_{12}}_b(c)), \sigma^{z_{23}}_{\tau^{z_{13}}_{\sigma^{z_{12}}_b\left(c\right)}\left(a\right)}
    (\tau^{z_{12}}_c\left(b\right)),  \tau^{z_{23}}_{\tau^{z_{12}}_c\left(b\right)}(\tau^{z_{13}}_{\sigma^{z_{12}}_b\left(c\right)}(a))). \label{rhs2} \end{equation}
 By equating (\ref{lhs2}) and (\ref{rhs2}) we arrive at (\ref{1})-(\ref{3}).  And conversely,  if conditions (\ref{1})-(\ref{3}) are satisfied then $R^{z_{ij}}$ automatically satisfies the parametric Yang-Baxter equation.
\end{proof}

\subsection{Generalized solutions from Drinfel'd twists}
In this subsection we construct generic solutions of the parametric set-theoretic Yang-Baxter equation by suitably twisting $p$-shelf solutions.
 We  first  introduce the notion of a parametric {\it Drinfel'd twist} and extend some earlier 
results shown for the non-parametric Yang-Baxter equation
 \cite{Doikou1, DoGhVl, DoiRyb22, DoRySt} to the parametric case. We note that a non-local twist type transformation for the non-parametric case was first introduced in \cite{Sol} and then further studied and exploited in \cite{Lebed, LebVen}.

\begin{defn}\label{def:Drinfiso}
Let $(X, R^{z_{ij}})$ and $(X,S^{z_{ij}})$ be solutions of the parametric set-theoretic Yang-Baxter equation.  We say that a map 
$\varphi^{z_{ij}}: X\times X\to X\times X$ 
is a \emph{Drinfel'd twist} (\emph{D-twist}) if for all $z_{i,j} \in Y,$
$$
\varphi^{z_{ij}}\, R^{z_{ij}} = S^{z_{ij}}\, (\varphi^{z_{ji}})^{(op)},
$$
where $(\varphi^{z_{ji}})^{(op)} = \pi \circ \varphi^{z_{ji}},$ and $\pi: X \times X \to X \times X$ is the ``flip'' map, such that for all $x,y\in X,$ $\pi(x,y) = (y,x).$
If $\varphi^{z_{ij}}$ is a bijection we say that $(X,R^{z_{ij}})$ 
and $(X, S^{z_{ij}})$ are {\it $D$-equivalent (via $\varphi^{z_{ij}}$)}, and 
we denote it by $R^{z_{ij}}\cong_D S^{z_{ij}}$.
\end{defn}
\begin{pro}\label{Th.r->r'}
    Let $(X,R^{z_{ij}})$ be a left non-degenerate solution, such that for all $a,b \in X,$ $z_{i,j} \in Y,$ $R^{z_{ij}}(b,a) = (\sigma^{z_{ij}}_{a}(b), \tau^{z_{ij}}_{b}(a))$ and let $(X,S^{z_{ij}})$ be a solution, such that for all $a,b \in X,$ $z_{i,j} \in Y,$  $S^{z_{ij}}(b,a) = (b ,b \triangleright_{z_{ij}} a)$ and $\tau_b^{z_{ij}}(a) : =(\sigma^{z_{ji}}_{\sigma^{z_{ij}}_a(b)})^{-1}(\sigma_a^{z_{ij}}(b)\triangleright_{z_{ij}} a).$ 
Then $R^{z_{ij}}$ is $D$-equivalent to $S^{z_{ij}}$.
\end{pro}
\begin{proof}
Let $\varphi^{z_{ij}}:X\times X\to X\times X$ be the map defined by
$~\varphi^{z_{ij}}(a,b):=(a,\sigma^{z_{ji}}_a(b)),$ 
for all $a,b\in X,$ $z_{i,j} \in Y$. 
$R^{z_{ij}}$ is left non-degenerate, hence $\varphi^{z_{ij}}$ is bijective and $(\varphi^{z_{ij}})^{-1}\left(a,b\right) = 
\left(a,(\sigma^{z_{ji}}_a)^{-1}\left(b\right)\right)$, also $(\varphi^{z_{ji}})^{(op)} (b,a) = (\sigma_{a}^{z_{ij}}(b), a)$ for all $a,b\in X$. Then 
\begin{eqnarray}
&&(\varphi^{z_{ij}})^{-1} S^{z_{ij}} (\varphi^{z_{ji}})^{(op)}(b,a)=(\varphi^{z_{ji}})^{-1} S^{z_{ij}} (\sigma^{z_{ij}}(b), a)=\nonumber\\
&&
(\varphi^{z_{ji}})^{-1}(\sigma^{z_{ij}}_a(b), \sigma^{z_{ij}}(b)\triangleright_{z_{ij}} a)= (\sigma^{z_{ij}}_a(b), (\sigma^{z_{ji}}_{\sigma^{z_{ij}}_a(b)})^{-1}(\sigma^{z_{ij}}(b)\triangleright_{z_{ij}} a)) \nonumber\\
&&  (\sigma^{z_{ij}}_a(b), \tau_{b}^{z_{ij}}(b)) = R^{z_{ij}}(b,a),\nonumber
\end{eqnarray}
where we have defined $\tau_b^{z_{ij}}(a) : =(\sigma^{z_{ji}}_{\sigma^{z_{ij}}_a(b)})^{-1}(\sigma_a^{z_{ij}}(b)\triangleright_{z_{ij}} a).$
That is $R^{z_{ij}}\cong_D S^{z_{ij}}$.
\end{proof}

\begin{rem} \label{rem1}
In the special case of reversible $R$-matrices we observe from the fundamental relation $R_{21}^{z_{ji}}R_{12}^{z_{ij}} = \id,$ that $\sigma^{z_{ji}}_{\sigma^{z_{ij}}_{a}{(b)}}(\tau_b^{z_{ij}}(a)) = a,$ which leads to $b\triangleright_{z_{ij}} a =a,$ and hence $S^{z_{ij}}(b,a) = (b,a)$ for all $a,b\in X,$ $z_{i,j}\in Y,$ i.e. $S^{z_{ij}} = \id.$ 
\end{rem}

\begin{defn}\label{def:twist:shelf}
    Let $(X,\triangleright_{z_{ij}})$ be a $p$-shelf.  We say that the twist $\varphi^{z_{ij}}: X \times X\to X \times X,$ such that 
for all $a,b\in X,$ $z_{i,j} \in Y,$ $\varphi^{z_{ij}}(a,b):=(a,\sigma^{z_{ji}}_a(b))$ and $\sigma_a^{z_{ij}}$ is a bijection, is an admissible twist, if for all $a,b, c\in X$, $z_{i,j,k } \in Y:$
    (1) $\sigma^{z_{ik}}_a(\sigma^{z_{ij}}_b(c)) = \sigma^{z_{ij}}_{\sigma^{z_{jk}}_a\left(b\right)}(\sigma^{z_{ik}}_{\tau^{z_{jk}}_b\left(a\right)}(c))$
    and (2) $\sigma^{z_{ik}}_c(b) \triangleright_{z_{ij}} 
    \sigma^{z_{jk}}_{c}(a) = \sigma^{z_{jk}}_c(b \triangleright_{z_{ij}} a).$
\end{defn}

In the following theorem we show that any left non-degenerate solution  $\left(X,\, R^{z_{ij}}\right)$ 
can be expressed in terms of the $p$-shelf $\left(X,\,\triangleright_{z_{ij}}\right)$  and its admissible twist.

\begin{thm}\label{le:lndsol}
    Let $\left(X,\,\triangleright_{z_{ij}} \right)$ be a $p$-shelf and $\varphi^{z_{ij}}: X \times X\to X \times X,$ such that $\varphi^{z_{ij}}(a,b):=(a,\sigma^{z_{ji}}_a(b))$ for all $a,b\in X,$ $z_{i, j} \in Y.$ Then, the map
$R^{z_{ij}}:X\times X\to X\times X$ defined by
    \begin{equation}\label{prop:solform}
        R^{z_{ij}}\left(b, a\right)  
        = \left(\sigma^{z_{ij}}_a\left(b\right),\, (\sigma^{z_{ji}}_{\sigma_a^{z_{ij}}(b)})^{-1}
        (\sigma_a^{z_{ij}}(b)\triangleright_{z_{ij}} a)\right),  
    \end{equation}
      for all $a,b\in X,$ $z_{i,j} \in Y$ is a solution if and only if $\varphi^{z_{ij}}$ is an admissible twist. 
\end{thm}
\begin{proof}
    We first assume that $\varphi^{z_{ij}}$ is an admissible twist. Then (\ref{1}) immediately follows from (1).\\
    We set
    $\tau_b^{z_{ij}}(a) : =(\sigma^{z_{ji}}_{\sigma^{z_{ij}}_a(b)})^{-1}(\sigma_a^{z_{ij}}(b)\triangleright_{z_{ij}} a),$ for all $a,b \in X,$ $z_{i,j}\in Y$. For $a,b,c\in X$, $z_{i,j,k} \in Y$ we have: 
     $$
    \begin{aligned}
\tau^{z_{ij}}_{\sigma^{z_{ik}}_{\tau^{z_{jk}}_b
\left(a\right)}\left(c\right)}(\sigma^{z_{jk}}_a\left(b\right))
    &= (\sigma^{z_{ji}}_{\sigma^{z_{ij}}_{\sigma^{z_{jk}}_a\left(b\right)}(\sigma^{z_{ik}}_{\tau^{z_{jk}}_b\left(a\right)}\left(c\right))})^{-1} (\sigma^{z_{ij}}_{\sigma^{z_{jk}}_a\left(b\right)}(\sigma^{z_{ik}}_{\tau^{z_{jk}}_b\left(a\right)}\left(c\right)) \triangleright_{z_{ij}}\sigma_a^{z_{jk}}(b))\\
&=(\sigma^{z_{ji}}_{\sigma^{z_{ik}}_{a}
(\sigma^{z_{ij}}_b(c))})^{-1} 
(\sigma^{z_{ik}}_{a}(\sigma^{z_{ij}}_b(c)) \triangleright_{z_{ij}}\sigma_a^{z_{jk}}(b))& \mbox{by (1) {Definition} \ref{def:twist:shelf}}\\
    &= (\sigma^{z_{ji}}_{\sigma^{z_{ik}}_{a}
(\sigma^{z_{ij}}_b(c))})^{-1} 
(\sigma_{a}^{z_{jk}}(\sigma_b^{z_{ij}}(c) \triangleright_{z_{ij}} b))
    &\mbox{by (2) {Definition} \ref{def:twist:shelf}}\\
&=(\sigma^{z_{ji}}_{\sigma^{z_{ik}}_{a}
(\sigma^{z_{ij}}_b(c))})^{-1} 
\big (\sigma_{a}^{z_{jk}}(\sigma^{z_{ji}}_{\sigma_b^{z_{ij}}(c)}(\tau^{z_{ij}}_c(b)))\big )& \\
    &= (\sigma^{z_{ji}}_{\sigma^{z_{ik}}_{a}
(\sigma^{z_{ij}}_b(c))})^{-1} 
\big (\sigma^{z_{ji}}_{\sigma^{z_{ik}}_{a}
(\sigma^{z_{ij}}_b(c))}(\sigma^{z_{jk}}_{\tau^{z_{ik}}_{\sigma^{z_{ij}}_b\left(c\right)}\left(a\right)}
    (\tau^{z_{ij}}_c\left(b\right)))\big )  & \mbox{by (1) {Definition} \ref{def:twist:shelf}} \\
    & = \sigma^{z_{jk}}_{\tau^{z_{ik}}_{\sigma^{z_{ij}}_b\left(c\right)}\left(a\right)}
    (\tau^{z_{ij}}_c\left(b\right)),    \end{aligned}
    $$
hence (\ref{3}) holds.

We now show (\ref{2}) using repeatedly the definition
$\tau_b^{z_{ij}}(a) : =(\sigma^{z_{ji}}_{\sigma^{z_{ij}}_a(b)})^{-1}(\sigma_a^{z_{ij}}(b)\triangleright_{z_{ij}} a),$ for all $a,b \in X,$ $z_{i,j}\in Y$:
$$
\begin{aligned}
&\tau^{z_{jk}}_{\tau^{z_{ij}}_c
\left(b\right)}(\tau^{z_{ik}}_{\sigma^{z_{ij}}_b\left(c\right)}\left(a\right))=
(\sigma^{z_{kj}}_{\sigma^{z_{jk}}_{\tau^{z_{ik}}_{\sigma^{z_{ij}}_b\left(c\right)}\left(a\right)}(\tau^{z_{jk}}_c\left(b\right))})^{-1}
(\sigma^{z_{jk}}_{\tau^{z_{ik}}_{\sigma^{z_{ij}}_b\left(c\right)}\left(a\right)}(\tau^{z_{ij}}_c\left(b\right)) \triangleright_{z_{jk}} \tau^{z_{ik}}_{\sigma^{z_{ij}}_b\left(c\right)}\left(a\right)) 
\\
&= (\sigma^{z_{kj}}_{\tau^{z_{ij}}_{\sigma^{z_{ik}}_{\tau^{z_{jk}}_b\left(a\right)}\left(c\right)}(\sigma^{z_{jk}}_a\left(b\right))})^{-1}
\Big ((\sigma^{z_{ji}}_{\sigma^{z_{ik}}_a(\sigma^{z_{ij}}_b\left(c\right))})^{-1}\big(\sigma^{z_{ik}}_a(\sigma^{z_{ij}}_b(c))
\triangleright_{z_{ij}} b\big) \triangleright_{z_{jk}}
(\sigma^{z_{ki}}_{\sigma^{z_{ik}}_a(\sigma^{z_{ij}}_b\left(c\right))})^{-1}\big(\sigma^{z_{ik}}_a(\sigma^{z_{ij}}_b\left(c\right)))
\triangleright_{z_{ik}}a\big )\Big )
 \\ 
 &  \mbox{by (\ref{1}), (\ref{3})} 
 \\
&= (\sigma^{z_{kj}}_{\tau^{z_{ij}}_{\sigma^{z_{ik}}_{\tau^{z_{jk}}_b\left(a\right)}\left(c\right)}(\sigma^{z_{jk}}_a\left(b\right))})^{-1}
\Big ((\sigma^{z_{ji}}_{\sigma^{z_{ik}}_a(\sigma^{z_{ij}}_b\left(c\right))})^{-1}\big( \sigma^{z_{ik}}_a(\sigma^{z_{ij}}_b\left(c\right)) \triangleright_{z_{ik}} (\sigma_a^{z_{jk}}(b) \triangleright_{z_{jk}} a) \big ) \Big)\\
& \mbox{by the $p$-self-distributivity and (2) of Definition \ref{def:twist:shelf}},
\\
&= (\sigma^{z_{ki}}_{\sigma^{z_{ik}}_{\tau^{z_{jk}}_b\left(a\right)}\left(c\right)})^{-1}
\Big((\sigma^{z_{ij}}_{\sigma^{z_{jk}}_a\left(b\right)})^{-1} \big (\sigma^{z_{ik}}_a(\sigma^{z_{ij}}_b\left(c\right))\big) \triangleright_{z_{ik}} \tau^{z_{jk}}_{b}(a)\Big ) \\
&
\mbox{by the $p$-self-distributivity, the definition of $\tau^{z_{ij}}_b(a)$ and (2) from Definition \ref{def:twist:shelf}},\\
&=  (\sigma^{z_{ki}}_{\sigma^{z_{ik}}_{\tau^{z_{jk}}_b\left(a\right)}\left(c\right)})^{-1}
\Big((\sigma^{z_{ij}}_{\sigma^{z_{jk}}_a\left(b\right)})^{-1} \big (\sigma^{z_{ij}}_{\sigma_a^{z_{jk}}(b)}(\sigma^{z_{ik}}_{\tau^{z_{jk}}_b(a)}\left(c\right))\big) \triangleright_{z_{ik}} \tau^{z_{jk}}_{b}(a)\Big ) \\
& \mbox{by (2) from Definition \ref{def:twist:shelf}},\\
&= (\sigma^{z_{ki}}_{\sigma^{z_{ik}}_{\tau^{z_{jk}}_b\left(a\right)}\left(c\right)})^{-1}
\Big(\sigma^{z_{ik}}_{\tau^{z_{jk}}_b(a)}\left(c\right)) \triangleright_{z_{ik}} \tau^{z_{jk}}_{b}(a)\Big )
    =\tau^{z_{ik}}_{c}(\tau^{z_{jk}}_{b}\left(a\right)),
    \end{aligned}
$$
and we conclude that (\ref{2}) is satisfied.

Conversely, let the map $R^{z_{ij}}$ be a solution on the set $X$. Then, condition (\ref{1})
coincides with the identity (1) of Definition \ref{def:twist:shelf}. We now show that, for all $a,b,c\in X$, $z_{i,j,k} \in Y$ identity (2) of Definition \ref{def:twist:shelf} holds. Indeed,
    $$
    \begin{aligned}
       \sigma^{z_{jk}}_c\left(b\triangleright_{z_{ij}} a\right)
        &= \sigma^{z_{ji}}_{\sigma^{z_{ik}}_c\left(b\right)}
\big (\sigma^{z_{jk}}_{\tau^{z_{ik}}_b\left(c\right)}(\tau^{z_{ij}}_{(\sigma_a^{z_{ij}})^{-1}\left(b\right)}\left(a\right))\big) & \mbox{by (\ref{1})}\\
        &= \sigma^{z_{ji}}_{\sigma^{z_{ik}}_c\left(b\right)}\big ( \sigma^{z_{jk}}_{\tau^{z_{ik}}_{\sigma^{z_{ij}}_a((\sigma_a^{z_{ij}})^{-1}\left(b\right))}\left(c\right)} \
(\tau^{z_{ij}}_{(\sigma^{z_{ij}}_a)^{-1}\left(b\right)}\left(a\right))\big)\\
        &= \sigma^{z_{ji}}_{\sigma^{z_{ik}}_c\left(b\right)}\big (\tau^{z_{ij}}_{\sigma^{z_{ik}}_{\tau^{z_{jk}}_a\left(c\right)}((\sigma^{z_{ij}}_a)^{-1}\left(b\right))}(\sigma^{z_{jk}}_c
\left(a\right))\big )& \mbox{by (\ref{3})}\\
        &= \sigma^{z_{ji}}_{\sigma^{z_{ik}}_c\left(b\right)}\big (
    \tau^{{z_{ij}}}_{(\sigma^{z_{ij}}_{\sigma^{z_{jk}}_c\left(a\right)})^{-1}(\sigma^{z_{ik}}_c\left(b\right))}(\sigma^{z_{jk}}_c\left(a\right)) \big ),&\\
        &= \sigma^{z_{ik}}_c\left(b\right)\triangleright_{z_{ij}} \sigma^{z_{jk}}_c\left(a\right) & \mbox{by (\ref{1})} \
    \end{aligned}
    $$

and this concludes our proof.
\end{proof}

\begin{cor} 
Any left non-degenerate solution $R^{z_{ij}}: X \times X \to X \times X,$\\ $R^{z_{ij}}(b,a)= (\sigma_a^{z_{ij}}(b), \tau^{z_{ij}}_b(a)),$ 
for all $a,b \in X,$ $z_{i,j} \in Y,$ can be obtained from a $p$-shelf solution,
where $a \triangleright_{z_{ij}} b =\sigma^{z_{ji}}_a(\tau^{z_{ij}}_{(\sigma^{z_{ij}}_b)^{-1}\left(a\right)}\left(b\right)),$ via an admissible twist.
\end{cor}
\begin{proof}
    To show that any left non-degenerate solution 
can be obtained from a $p$-shelf solution, it is enough to show that $\sigma_a^{z_{ij}}$ satisfies the properties of Definition \ref{def:twist:shelf}, where $a\triangleright_{z_{ij}} b:= \sigma^{z_{ji}}_a(\tau^{z_{ij}}_{(\sigma^{z_{ij}}_b)^{-1}\left(a\right)}\left(b\right)).$
These  follow from the proof of Theorem \ref{le:lndsol} (specifically the second part of the proof).
\end{proof}

We conclude, given the findings of this subsection, that the problem of finding generic solutions of the parametric set-theoretic Yang-Baxter equation is reduced to the classification of $p$-shelf/rack solutions and the identification of admissible twists.

\begin{cor}\label{cor:bij-lnd}
A left non-degenerate solution $(X,R^{z_{ij}})$ is bijective if and only if $(X,\triangleright_{z_{ij}})$ is a $p$-rack for all $z_{i,j} \in Y$.
\end{cor}
\begin{proof}
 This follows from the fact that $R^{z_{ij}}$ is invertible if and only if the $p$-shelf solution is invertible, 
 i.e. it is a $p$-rack solution. 
\end{proof}

Recall Remark \ref{rem1}, which together with Theorem \ref{le:lndsol} 
show that all reversible solutions are obtained from the identity map via an admissible twist.

\begin{pro} \label{prora2}
 Let $(X, +,\circ)$ be a skew brace and let $Y\subseteq X,$ be such that  
 \begin{itemize}
   \item for all $w, z\in Y,$ $z\circ w = w \circ z,$ 
   \item for all $a,b \in X,$ $z\in Y,$ $(a+b)\circ z = a \circ z -z +b \circ z,$
   \item $z \in Y$ are central in $(X, +).$
    \end{itemize}
    Let also $\varphi^{z_{ij}}: X \times X \to X \times X,$ be such that for all $a,b\in X,$ $z_{i,j}, \in Y$ and some $\xi \in Y,$ that is also central in $(X, \circ),$ $\varphi^{z_{ij}}(a,b) = (a,\sigma^{z_{ji}}_a(b)),$ 
 where $\sigma^{z_{ij}}_a(b) = z_i^{-1} -\xi\circ a \circ z_i^{-1} \circ z_j + a \circ b \circ\xi \circ z_j.$
   We also define for all $a,b \in X,$ $z_{i,j} \in Y$ and $\xi \in Y,$ $\triangleright_{z_{ij}}: X \times X \to X,$ such that $a \triangleright_{z_{ij}}b := -\xi\circ a \circ z_i\circ z_j^{-1} +\xi \circ b + a\circ z_i \circ z_j^{-1}$ and $\tau_b^{z_{ij}}: X \to X,$ such that  $\tau_{b}^{z_{ij}}(a) = (\sigma^{z_{ji}}_{\sigma^{z_{ij}}_a(b)})^{-1}\big (\sigma^{z_{ij}}_a(b) \triangleright_{z_{ij}}a \big).$ Then:
    \begin{enumerate}
        \item For all $a,b\in X$ and $z_{i,j}\in Y,$ 
        $ a\circ b = \sigma_a^{z_{ij}}(b)\circ \tau_b^{z_{ij}}(a).$ 
      
        \item $\varphi^{z_{ij}}$ is an admissible twist, for all $z_{i,j} \in Y$.    
        \end{enumerate}
\end{pro}
\begin{proof}
Let us first observe  that we can identify 
    the inverse map $(\sigma^{z_{ij}}_b)^{-1},$  i.e. $\sigma^{z_{ij}}_b$ is a bijection.  It follows from $~\sigma^{z_{ij}}_a((\sigma^{z_{ij}}_a)^{-1}(b)) = b$ that $~(\sigma^{z_{ij}}_a)^{-1}(b)= z_i^{-1} - a^{-1}\circ \xi^{-1} \circ z_i^{-1}\circ z_j^{-1} + a^{-1}\circ b\circ \xi^{-1} \circ z_j^{-1},$ for all $a,b\in X,$ $z_{i,j}\in Y$ and $\xi \in Y.$    
    \begin{enumerate}
   \item  We may now extract the map $\tau^{z_{ij}}_b:$   
    \begin{eqnarray}
   \tau_{b}^{z_{ij}}(a) &=& (\sigma^{z_{ji}}_{\sigma^{z_{ij}}_a(b)})^{-1}\big (\sigma^{z_{ij}}_a(b) \triangleright_{z_{ij}}a \big)  \nonumber\\ &= & z_j^{-1} - (\sigma^{z_{ij}}_{a}(b))^{-1}\circ \xi^{-1} \circ z_j^{-1} \circ z_i^{-1} +(\sigma^{z_{ij}}_{a}(b))^{-1}\circ  (\sigma^{z_{ij}}_a(b) \triangleright_{z_{ij}}a)\circ \xi^{-1} \circ z_i^{-1} \nonumber \\
   &= &  (\sigma_a^{z_{ij}}(b))^{-1}\circ a \circ z_i^{-1} - (\sigma_a^{z_{ij}}(b))^{-1}\circ \xi^{-1}\circ z_i^{-1} \circ z_j^{-1} + \xi^{-1} \circ z_j^{-1}\nonumber  
    \end{eqnarray}
    We now directly compute for all $a,b\in X$ and $z_{i,j}\in Y,$ 
    \begin{eqnarray}
    \sigma_a^{z_{ij}}(b)\circ \tau_b^{z_{ij}}(a) &=&  \sigma_a^{z_{ij}}(b) \circ  \big ((\sigma_a^{z_{ij}}(b))^{-1}\circ a \circ z_i^{-1} - (\sigma_a^{z_{ij}}(b))^{-1}\circ \xi^{-1}\circ z_i^{-1} \circ z_j^{-1} +\xi^{-1}\circ z_j^{-1}\big )  \nonumber\\ &=&  a\circ z_i^{-1}-
    \xi^{-1}\circ z_{i}^{-1}\circ z_j^{-1}+\sigma_a^{z_{ij}}(b)\circ \xi^{-1}\circ  z_j^{-1}  \nonumber \\
    &=&  a\circ z_i^{-1}-
    \xi^{-1} \circ z_{i}^{-1}\circ z_j^{-1}+ \big (z_i^{-1} - a\circ \xi\circ z_i^{-1}\circ z_j + a\circ b\circ \xi\circ z_j \big )\circ z_j^{-1}\circ \xi^{-1}\nonumber \\
    &=&  a\circ b. \nonumber \end{eqnarray}
   \item  To prove that $\varphi^{z_{ij}}$ is an admissible twist we need to show properties (1) and (2) of Definition 
    \ref{def:twist:shelf}.  Indeed,
    \begin{eqnarray}
        \sigma^{z_{ik}}_a(\sigma_b^{z_{ij}}(c)) &=& 
z_i^{-1} - a \circ \xi \circ z_i^{-1} \circ z_k + 
        a \circ \sigma^{z_{ij}}_b(c) \circ \xi \circ z_k  \nonumber\\
        &=&  z_i^{-1} - a \circ b \circ z_i^{-1}\circ z_j \circ z_k\circ \xi \circ \xi + a\circ b \circ c\circ z_j \circ z_k\circ \xi \circ \xi
\nonumber\\
        &=&  z_i^{-1} - \sigma_a^{z_{jk}} \circ \tau^{z_{jk}}_b(a) \circ z_i^{-1}\circ z_j \circ z_k\circ \xi \circ \xi + \sigma_a^{z_{jk}} \circ \tau^{z_{jk}}_b(a) \circ c\circ z_j \circ z_k\circ \xi \circ \xi , \nonumber\\
        &=& \sigma^{z_{ij}}_{\sigma_a^{z_{jk}}(b)}(\sigma_{\tau^{z_{jk}}_b(a)}^{z_{ik}}(c)),  \nonumber  \end{eqnarray}
i.e.  we conclude that conditions (1) of Definition \ref{def:twist:shelf} hold; we also observe that\\ $\sigma_a^{z_{ik}}(\sigma_b^{z_{ij}}(c)) = \sigma_a^{z_{ij}}(\sigma_b^{z_{ik}}(c)).$
Now we show condition (2) of Definition \ref{def:twist:shelf}. We first directly compute
\begin{eqnarray}
\sigma^{z_{jk}}_c(b \triangleright_{z_{ij}} a) &=& z_j^{-1}  - c \circ \xi \circ z_j^{-1} \circ z_k \nonumber\\ & & + c \circ (z_j^{-1} - \xi\circ b\circ z_i\circ z_j^{-1} + \xi\circ a - z_j^{-1} + b\circ z_i \circ z_j^{-1} )\circ \xi \circ z_k \nonumber\\ & =& z_j^{-1} - c\circ b \circ \xi \circ \xi \circ z_i \circ z_j^{-1} \circ z_k + c\circ a \circ \xi \circ \xi \circ z_k \nonumber\\ & &
- c\circ \xi  \circ z_j^{-1} \circ z_k + c \circ b \circ \xi  \circ z_i \circ z_j^{-1} \circ z_k. \nonumber
\end{eqnarray}
Similarly, we compute
\begin{eqnarray}
\sigma_c^{z_{ik}}(b) \triangleright_{z_{ij}} \sigma^{z_{jk}}_c(a) & =& -\xi\circ \sigma_c^{z_{ik}}(b)\circ z_i \circ z_j^{-1}+ \xi \circ \sigma_c^{z_{jk}}(a) + \sigma_c^{z_{ik}}(b) \circ z_i \circ z_j^{-1}\nonumber \\ 
&= &
z_j^{-1} - c\circ b \circ z_i \circ z_j^{-1} \circ z_k + c\circ a \circ z_k - c \circ z_j^{-1} \circ z_k + c \circ b \circ z_i \circ z_j^{-1} \circ z_k. \nonumber \\
& =& z_j^{-1} - c\circ b \circ \xi \circ \xi \circ z_i \circ z_j^{-1} \circ z_k + c\circ a \circ \xi \circ \xi \circ z_k \nonumber\\ & &
- c\circ \xi  \circ z_j^{-1} \circ z_k + c \circ b \circ \xi  \circ z_i \circ z_j^{-1} \circ z_k. \nonumber
\end{eqnarray}
Comparing the two expression above we conclude that condition (2) of Definition \ref{def:twist:shelf} also holds, 
i.e.  $\varphi^{z_{ij}}$ is indeed an admissible twist.
\hfill \qedhere
    \end{enumerate}
    \end{proof}
Example \ref{ex:cyclicbraces} can be used for the construction of $\sigma_a^{z_{ij}}(b)$ and $\tau^{z_{ij}}_b(a)$ of Proposition \ref{prora2} (see also Example \ref{good1}).
\begin{rem} \label{rem0b} 
If $(X, +, \circ)$ is a brace, i.e. $(X, +)$ is an abelian group and $\xi =1,$ then for all $a,b \in X,$ 
$z_{i,j} \in Y,$ 
$a \triangleright_{z_{ij}} b =  -a\circ z_i \circ z_j^{-1}+b + a \circ z_i \circ z_j^{-1} = b,$ hence $S^{z_{ij}}(a,b) = (a, a \triangleright_{z_{ij}} b) =(a,b) $,  i.e. $ S^{z_{ij}} = \id.$
Also,
$\sigma^{z_{ji}}_{\sigma^{z_{ij}}_a(b)}(\tau^{z_{ij}}_{b}(a)) =a,$ hence the map $R^{z_{ij}}:X \times X \to X \times X,$  $R^{z_{ij}}(b,a) = ( \sigma^{z_{ij}}_a(b), \tau^{z_{ij}}_{b}(a) )$ obtained from $S^{z_{ij}}$ via the admissible twist,  is reversible,  i.e.  $R_{12}^{z_{ij}} R_{21}^{z_{ji}}= \id.$ 
\end{rem}

\subsection{Parametric Yang-Baxter operators}

This {subsection provides a key motivation for the material} presented in the subsequent section given that the
generalized algebraic structures studied here encapsulate part of the underlying set-theoretic Yang-Baxter algebras, which
are introduced in the next section.
Bearing in mind the definition of braided groups and braidings in \cite{chin}, \cite{GatMaj} 
and their deformations \cite{DoiRyb22, DoRy23, DoRySt} we further
generalize the definitions to introduce the parametric Yang-Baxter structures.

It is useful for the purposes of this subsection to introduce the following. Let $X$ and 
$Y \subseteq X$ be non-empty sets and introduce for all 
$z_{i, j, k} \in Y$ the maps, 
$M^{z_{ijk}}_y \in \{f^{z_{ijk}}_y, g^{z_{ijk}}_{y}, \hat f^{z_{ijk}}_y, \hat g^{z_{ijk}}_{y} \},$ $M_y^{z_{ijk}}: X \to X,$
$x \mapsto M^{z_{z_{ijk}}}_y(x),$ and the maps $S^{z_{ij}}_y \in \{\sigma_y^{z_{ij}}, \tau_y^{z_{ij}}\},$ $S_y^{z_{ij}}: X \to X,$ $x \mapsto S^{z_{ij}}_y(x).$

\begin{defn} \label{def0} 
Let $(X,\circ)$ be a group, $Y\subseteq X$ and consider the following maps for $z_{i,j, k} \in Y,$ $m: X \times X \to X,\ (x,y)\mapsto x \circ y,$
$~\pi: X \times X \to X \times X,\ (x,y) \mapsto (y,x),$ and   $R^{z_{ij}}, \xi^{z_{ijk}}, \zeta^{z_{ijk}}: X \times X \to X \times X:$
\[R^{z_{ij}}(y,x)= (\sigma^{z_{ij}}_x(y), \tau^{{z}_{ij}}_y(x)), ~~\xi^{z_{ijk}}(y, x) = (f^{z_{ijk}}_x(y),  g_y^{z_{ijk}}(x)),  ~~\zeta^{z_{ijk}}(y,x)= (\hat f_{x}^{z_{ijk}}(y), \hat g^{z_{ijk}}_y(x)),\] such that for all $x, y \in X,$ $z_{i,j,k} \in Y,$ $f^{z_{ijk}}_x(y) =f^{z_{ikj}}_x(y),$ $g^{z_{ijk}}_y(x) = g^{z_{ikj}}_y(x),$ $\hat f^{z_{ijk}}_x(y) = \hat f^{z_{jik}}_x(y),$ $\hat g^{z_{ijk}}_y(x) = \hat g^{z_{jik}}_y(x)$ and 
$M_y^{z_{ijk}}, S_y^{z_{ij}}$ are bijections. Let also
$\hat m := m \pi.$ The map $R^{z_{ij}}$ is called a parametric ($p$)-set Yang-Baxter operator, and the group is called a parametric ($p$)-set Yang-Baxter group, if for all $x,y,w \in X,$ 
$z_{i,j,k} \in Y:$
\begin{enumerate}
\item $\hat m (y, x) =m (R^{z_{ij}} (y, x)).$
\item $ \xi^{z_{ijk}} (\id_{{X}} \times \hat m ) (w,y,x)= (\id_X \times \hat  m) R_{13}^{z_{ik}} R_{12}^{z_{ij}} (w,y,x).$
\item 
$\zeta^{z_{ijk}}(\hat m \times \id_X )(w,y,x)= (\hat m \times \id_X)  R_{13}^{z_{ik}}R_{23}^{z_{jk}}(w,y,x).$
\end{enumerate}
\end{defn}

\begin{pro} \label{lemmA} 
Let $(X,\circ)$ be a $p$-set Yang-Baxter group and 
the map $R^{z_{ij}}: X \times X \to X \times X$ be a $p$-set Yang-Baxter operator, then: 

\begin{enumerate}
\item $R^{z_{ij}}$ satisfies the parametric Yang-Baxter  equation.
\item $x \circ y= f^{z_{ijk}}_x(y) \circ g_y^{z_{ijk}}(x)=   \hat f_{x}^{z_{ijk}}(y) \circ \hat g^{z_{ijk}}_y(x).$ 
\end{enumerate}
\end{pro}
\begin{proof}
For the first part it suffices to show that conditions (\ref{1})-(\ref{3}) hold. 
\begin{enumerate}

\item Indeed, from part (2) of Definition \ref{def0} we obtain for all $x,y,w \in X,$ $z_{i,j,k} \in Y$:
\begin{equation}
\sigma_{x}^{z_{ik}}(\sigma_y^{z_{ij}}(w)) = f^{z_{ijk}}_{x \circ y}(w) =   f^{z_{ikj}}_{\sigma^{z_{jk}}_x(y) \circ \tau^{z_{jk}}_y(x)}(w)  =\sigma_{\sigma^{z_{jk}}_x(y)}^{z_{ij}}(\sigma_{ \tau^{z_{jk}}_y(x)}^{z_{ik}}(w)), \label{11b}
\end{equation}
hence condition  (\ref{1}) of the Yang-Baxter equation is satisfied.  

From part (3) of Definition \ref{def0} we obtain:
\begin{equation}
\tau^{z_{ik}}_w(\tau_y^{z_{jk}}(x)) = \hat g^{z_{ijk}}_{y\circ w}(x)  = \hat g^{z_{jik}}_{\sigma^{z_{ij}}_y(w)\circ \tau_w^{z_{ij}}(y)}(x) =  \tau^{z_{jk}}_{\tau_w^{z_{ij}}(y)}(\tau_{\sigma_y^{z_{ij}}(w)}^{z_{ik}}(x)), \label{4a}
\end{equation}
hence condition (\ref{2}) of the Yang-Baxter equation is also satisfied.  It remains now to show (\ref{3}). From part (3) of definition \ref{def0} we also have,
\begin{eqnarray}
\sigma^{z_{jk}}_x(y) \circ  \sigma^{z_{ik}}_{\tau^{z_{jk}}_y(x)}(w) &=& \hat f^{z_{ijk}}_x(y \circ w) =  \hat f^{z_{jik}}_x(\sigma^{z_{ij}}_y(w) \circ \tau^{z_{ij}}_w(y))    \label{4b}  \\ & =& \sigma^{z_{ik}}_x(\sigma_y^{z_{ij}}(w)) \circ  \sigma^{z_{jk}}_{\tau^{z_{ik}}_{\sigma_y^{z_{ij}}(w)}(x)}(\tau^{z_{ij}}_{w}(y)).  \label{4}
\end{eqnarray}

Consider the RHS of (\ref{3}), which via part (1) of Definition \ref{def0} becomes
\begin{eqnarray}
\tau^{z_{ij}}_{\sigma^{z_{ik}}_{\tau^{z_{jk}}_y\left(x\right)}\left(w\right)}(\sigma^{z_{jk}}_x\left(y\right) )  &=&  (\sigma^{z_{ik}}_x(\sigma_y^{z_{ij}}(w)))^{-1} \circ \sigma^{z_{jk}}_x(y) \circ  \sigma^{z_{ik}}_{\tau^{z_{jk}}_y(x)}(w), \qquad \mbox{by (\ref{4})} \nonumber\\
&=&  \sigma^{z_{jk}}_{\tau^{z_{ik}}_{\sigma_y^{z_{ij}}(w)}(x)}(\tau^{z_{ij}}_{w}(y))
\end{eqnarray}
and this is the proof of condition (\ref{3}). That is, we conclude that the $p$-set Yang-Baxter operator is a solution of the parametric Yang-Baxter equation.

\item From  (\ref{4a}), (\ref{4b}) and (1) of Definition \ref{def0} we conclude that $x \circ y=  \hat f_{x}^{z_{ijk}}(y) \circ \hat g^{z_{ijk}}_y(x).$ From (2) of Definition  \ref{def0} we also have that  
\[ \tau_{\sigma_y^{z_{ij}}(w)}^{z_{ik}} \circ \tau_{w}^{z_{ij}}(y)  = g^{z_{ijk}}_{w}(x\circ y).\]
The latter relation together with (\ref{11b}) and (1) of Definition \ref{def0} lead also to $x\circ y =f_{x}^{z_{ijk}}(y) \circ  g^{z_{ijk}}_y(x) .$
\hfill \qedhere
\end{enumerate}
\end{proof}

\begin{exa} \label{exa11}
Let $(X, +, \circ)$ be a skew brace and recall the map of Proposition \ref{prora2} for all $a\in X,$ $z_{i,j} \in Y,$ and some $\xi \in Y,$ $\sigma_a^{z_{ij}}: X \to X,$ $\sigma^{{z_{ij}}}_a(b) = z_i^{-1} -a\circ z_i^{-1} \circ z_j\circ \xi +a\circ b\circ z_j \circ \xi,$ $a,b \in X,$ which satisfies:
\begin{equation}
a\circ b = \sigma^{z_{ij}}_a(b) \circ \tau_b^{z_{ij},\xi}(a) \quad \mbox{and} \quad  \sigma^{z_{ik}}_a(\sigma^{z_{ij}}_b(c)) = \sigma^{z_{i}, z_j\circ z_k\circ \xi}_{a\circ b}(c),
\end{equation}
$f^{z_{ijk}}_a(b) = \sigma^{z_{i}, z_j\circ z_k \circ \xi}_{a}(b).$  Moreover, from (\ref{4b}) we obtain that $\hat f^{z_{ijk}}_a(b) = \sigma_a^{z_{i\circ j k}}(b),$ and $z_{i\circ j k}$ denotes dependence  on $(z_i\circ z_j, z_k).$
\end{exa}

\begin{defn} \label{defA}
Let $X$ and $Y \subseteq X$ be non-empty sets and let for all $z_{i,j} \in Y$ the binary operation 
$\bullet_{z_{ij}}: X \times X \to X,$ $(x, y) \mapsto x \bullet_{z_{ij}} y.$ 
Consider also for all $ z_{i,j} \in Y$ the following maps, $m_{z_{ij}}: X \times X \to X,\ (x,y)\mapsto x \bullet_{z_{ij}} y,$
$~\pi: X \times X \to X \times X,\ (x,y) \mapsto (y,x):$
\[R^{z_{ij}}(y,x) = (y,  \tau^{z_{ij}}_y(x)),\quad {\xi^{z_{ijk}}(y, x) = (y,  g_y^{z_{ijk}}(x)),}  \quad  \zeta^{z_{ijk}}(y,x) = (y, \hat g^{z_{ijk}}_y(x),)\]
such that $\tau_y^{z_{ij}},\ g_y^{z_{ijk}},\ \hat g_y^{z_{ijk}}$ are bijections for all $x, y {\in} X,$ $z_{i,j,k} \in Y$ and
{$g^{z_{ijk}}_y(x) = g^{z_{ikj}}_y(x),$} $\hat g^{z_{ijk}}_y(x) = \hat g^{z_{jik}}_y(x).$ Let also $\hat m_{z_{ij}} := m_{z_{ij}} \pi.$
The map $R^{z_{ij}}$ is called  a 
$p$-rack operator, and $(X, \bullet_{z_{ij}})$ is called a $p$-rack magma, if for all 
$x,y,w \in X,$ $z_{i,j,k} \in Y:$
\begin{enumerate}
\item $\hat m_{z_{ji}} (y, x) =m_{z_{ij}} (R^{z_{ij}} (y, x)).$
\item $\xi^{z_{ikj}} (\id \times \hat m_{kj} ) (w,y,x)= (\id_X \times \hat  m_{kj}) R_{13}^{z_{ik}} R_{12}^{z_{ij}} (w,y,x).$
\item 
$\zeta^{z_{ijk}}(\hat m_{z_{ji}} \times \id_X )(w,y,x)= (\hat m_{z_{ji}} \times \id_X)  R_{13}^{z_{ik}}R_{23}^{z_{jk}}(w,y,x).$
\end{enumerate}
\end{defn}
We note that Definition \ref{defA} could be seen as a special case of Definition 
\ref{def0}, however the underlying algebraic structure in \ref{defA} 
is a parametric one and is not necessarily a group as in \ref{def0}.

\begin{pro}\label{le:p-def1}
Let $(X,\bullet_{z_{ij}})$ be a $p$-rack magma and the map 
$R^{z_{ij}}: X \times X \to X \times X,$ such that $R^{z_{ij}}({y,x}) = (y,  \tau^{z_{ij}}_y(x))$ be a 
$p$-rack operator for all $z_{i,j}\in Y.$
Then 
$R^{z_{ij}}$ satisfies the parametric Yang-Baxter  equation.
\end{pro}
\begin{proof}
Let $\tau^{z_{ij}}_b(a) := b \triangleright_{z_{ij}}(a).$ Then, for all $x,y,w \in X,$ $z_{i,j,k} \in Y,$ we obtain by condition (1) of Definition \ref{defA}:
 $~x\bullet_{z_{ji}} y= y\bullet_{z_{ij}}(y\triangleright_{z_{ij}} x)$
   and from condition ${(2)}$:
\begin{eqnarray}
 w \triangleright_{z_{ik}} (y\triangleright_{z_{jk}}x) 
 =  \hat g^{z_{ijk}}_{y \bullet_{z_{ji}} w}(x) =\hat g^{z_{jik}}_{w \bullet_{z_{ij}} (w\triangleright_{z_{ij}} y)}(x)= (w\triangleright_{z_{ij}} y) \triangleright_{z_{jk}} (w\triangleright_{z_{ik}} x).\nonumber
\end{eqnarray}
The latter condition is the $p$ self-distributivity, i.e. $(X, \triangleright_{z_{ij}})$ is a $p$-rack (recall from Definition \ref{defA} $w\triangleright_{z_{ij}}$ is a bijection for all $w \in X,$ $z_{i,j} \in Y$) and hence according to Proposition \ref{prose} $R^{z_{ij}}$ is a solution of the Yang-Baxter equation. Due to the fact that $(X, \triangleright_{z_{ij}})$ is a $p$-rack the solution is invertible.
\end{proof}

\begin{exa} \label{bullet}
Let $(X, +, \circ)$ be a skew brace and recall the binary operation from Proposition \ref{prora2}, 
    $ \triangleright_{z_{ij}}: X\times X \to X,$ such that 
    and 
    $a\triangleright_{z_{ij}} b  = - \xi\circ a \circ z_{i} \circ z_j^{-1} + \xi\circ b + a\circ z_i \circ z_j^{-1},$ 
    for all $a,b\in X,$ $z_{i,j} \in Y$ and some $\xi \in X.$ Let also $\bullet_{z_{ij}}: X \times X\to X,$ such that $a\bullet_{z_{ij}}b =\xi \circ a \circ z_i + b \circ z_j,$   $a,b\in X,$ $z_{i,j} \in Y,$ and some $\xi \in X.$ Then the following holds:
    \[a \bullet_{z_{ji}} b = b \bullet_{z_{ij
}} (b \triangleright_{z_{ij}} a),\]
where $z_o =1.$ If $(X, +\circ)$ is a brace, the following relation also holds,
\[
b \triangleright_{z_{ik}} (a\triangleright_{z_{jk}} c) 
 = (a\bullet_{z_{ji}}b) \triangleright_{z_{ok}} (\xi\circ c).\]
Hence the function $\hat g^{z_{ijk}}_y(x)$ of Definition \ref{defA} is $\hat g^{z_{ijk}}_y(x)=: \hat g^{z_{ok}}_y(x) = y \triangleright_{z_{ok}} (\xi \circ x).$
\end{exa}

\noindent {\it Note.
The linearized version of the set-theoretic Yang-Baxter equation:}\\
Consider a vector space $V= \mathbb{C}X$ of dimension equal to the cardinality of $X$. 
Let  ${\mathbb B} = \{e_a\}_{a\in X}$ be the basis of $V$ and ${\mathbb B}^* = \{e^*_a\}_{a\in X}$ 
be the dual basis: $e_a^* e_b= \delta_{a,b },$ also  $e_{a,b} := e_a  e_b^*.$
Let also $R^{z_{ij}} \in \EEnd({\mathbb C}^n\otimes {\mathbb C}^n),$ 
$z_{i,j} \in Y,$ be a solution of the tensor (quantum) parametric Yang-Baxter equation 
\begin{equation}
R_{12}^{z_{12}} \ R^{z_{13}}_{13}\ R_{23}^{z_{23}} =  R_{23}^{z_{23}}\ R^{z_{13}}_{13}\  R_{12}^{z_{12}}, \label{TYBE}
\end{equation}
where in general $R^{z_{ij}} = \sum_{a,c,bd} R^{z_{ij}}(b,c| a,d) e_{b,c} \otimes e_{a,d}$ and
$~R_{12}^{z_{ij}} =\sum_{a,c,b,d} R^{z_{ij}}(b,c| a,d) e_{b,c} \otimes e_{a,d}\otimes 1_V,$ 
$~R_{13}^{z_{ij}}(c,b,a) =\sum_{a,c,b,d} R^{z_{ij}}(b,c| a,d) e_{b,c} \otimes 1_V \otimes e_{a,d},$ and\\ $~R_{23}^{z_{ij}}(c,b,a) = 1_V \otimes \sum_{a,c,b,d} R^{z_{ij}}(b,c| a,d) e_{b,c} \otimes e_{a,d},$ $z_{i,j} \in Y.$ 
 
 In the case of set-theoretic solutions specifically, $R^{z_{ij}}(b,c| a,d) = \delta(c, \sigma_{a}^{z_{ij}} (b)) \delta(d, \tau^{z_{ij}}_b(a)),$ whereas for $p$-shelves solutions, $R^{z_{ij}}(b,c| a,d)= \delta(c, b) \delta(d, b\triangleright_{z_{ij}} a).$
This description can be formally generalized to infinite countable sets, but for compact sets the  description requires the use of functional analysis and the study of kernels of integral operators that correspond to the solution $R^{z_{ij}}.$ $\square$

In the subsequent section we establish the algebraic framework in the tensor product formulation. This naturally provides solutions to the parametric set-theoretic Yang-Baxter equation,  thus the linearized version presented above is certainly relevant in what follows.

\section{Parametric set-theoretic algebras}

\noindent In this section we are exploring the underlying algebraic structures that provide  the universal ${\cal R}$-matrices associated to $p$-rack and set-theoretic solutions of the parametric Yang-Baxter equation.  We show that universal set-theoretic solutions are obtained from the $p$-shelf ones via suitable admissible universal  twists. 

We first introduce the parametric rack ($p$-rack)
and restricted $p$-rack algebras.  Using these algebras we are able to extract the associated universal ${\cal R}$-matrices and coproduct structures.  Certain fundamental representations of these algebras lead to set-theoretic solutions of the parametric Yang-Baxter equation.

\subsection{Parametric rack algebras}

\noindent We start our analysis by defining the algebra associated to 
$p$-shelf or rack solutions of the parametric Yang-Baxter equation called {\it parametric shelf or rack ($p$-shelf, rack) algebra}
(see also relevant findings for the non-parametric case \cite{DoRySt}).
\begin{defn} \label{qualgd0} ($p$-shelf algebra.) 
Let ${Y \subseteq X}$ be non-empty sets. We define for all $z_{i, j,k}\in Y,$ ($i,j,k\in {\mathbb Z}^+,$) 
the binary operation, $\triangleright_{z_{ij}}: X\times X \to X,$ $(a,b) \mapsto  a\triangleright_{z_{ij}}b.$
We say that the unital, associative algebra ${\cal Q},$ over a field $k$ 
generated by indeterminates,  
${1_\cal Q}$ (unit element) $q^{z_{ij}}_a, \ (q_a^{z_{ij}})^{-1},\ h_a\in {\cal Q}$ 
and relations for all $a,b \in X,$ $z_{i,j,k}\in Y:$
\begin{eqnarray}
q_a^{z_{ij}} (q_{a}^{z_{ij}})^{-1} =(q_{a}^{z_{ij}})^{-1}q_a^{z_{ij}} = 1_{\cal Q},\quad q_a^{z_{jk}} q^{z_{ik}}_b = q^{z_{ik}}_b  
q^{z_{jk}}_{b \triangleright_{z_{ij}} a}, \quad  h_a  h_b =\delta_{a, b} h_a^2, 
\quad q^{z_{ij}}_b  h_{b\triangleright_{z_{ij}} a}= h_a q^{z_{ij}}_b \label{qualg}
\end{eqnarray}
is a $p$-shelf algebra.
\end{defn}

The choice of the name $p$-shelf algebra is justified by the following statements.
\begin{pro}  \label{proro}
\label{qua1}
Let ${\cal Q}$ be a $p$-shelf algebra then for all $a,b,c \in X,$ $z_{i,j,k} \in Y,$ 
\[h_{c \triangleright_{z_{ik}}(b \triangleright_{z_{jk}}a)}
=h_{(c \triangleright_{z_{ij}}b) \triangleright_{z_{jk}}( c\triangleright_{z_{ik}}a)}.\]

If also for all $a,b \in X,$ $h_a = h_b \Rightarrow a =b,$ then for all $a,b,c \in X,$ $z_{i,j,k} \in Y,$ $c \triangleright_{z_{ik}}(b \triangleright_{z_{jk}}a) = (c \triangleright_{z_{ij}} b)\triangleright_{z_{jk}}( c\triangleright_{z_{ik}}a),$ i.e.  
$(X, \triangleright_{z_{ij}})$ is a $p$-shelf and also $a\triangleright_{z_{ij}}$ is injective.
\end{pro}
\begin{proof}
We compute $h_a q^{z_{jk}}_b q^{z_{ik}}_c$ 
using the associativity of the algebra:
\begin{eqnarray}
&& h_a q^{z_{jk}}_b q^{z_{ik}}_c =q^{z_{jk}}_b  h_{b\triangleright_{z_{jk}}a} q^{z_{ik}}_c = q^{z_{jk}}_b q^{z_{ik}}_c h_{c \triangleright_{z_{ik}}(b \triangleright_{z_{jk}}a)} = 
q^{z_{ik}}_c q^{z_{jk}}_{c \triangleright_{z_{ij}}b}
h_{c \triangleright_{z_{ik}}(b \triangleright_{z_{jk}}a)}, \label{w1}\\
&& h_a q^{z_{jk}}_b q^{z_{ik}}_c  =h_a q^{z_{ik}}_c q^{z_{jk}}_{c\triangleright_{z_{ij}}b}  = q_c^{z_{ik}} h_{c\triangleright_{z_{ik}}a} q^{z_{jk}}_{c \triangleright_{z_{ij}} b } = 
q^{z_{ik}}_c q^{z_{jk}}_{c \triangleright_{z_{ij}} b} h_{(c \triangleright _{z_{ij}} b)\triangleright_{z_{jk}}( c\triangleright_{z_{ik}}a)}. \label{w2}
\end{eqnarray}
Due to invertibility of $q^{z_{ij}}_a$ for all $a \in X$ we conclude from (\ref{w1}), (\ref{w2}) that 
\[h_{c \triangleright_{z_{ik}}(b \triangleright_{z_{jk}}a)}
=h_{(c \triangleright_{z_{ij}}b) \triangleright_{z_{jk}}( c\triangleright_{z_{ik}}a)}\  
 \Rightarrow\ c \triangleright_{z_{ik}}(b \triangleright_{z_{jk}}a) =(c \triangleright_{z_{ij}} b)\triangleright_{z_{jk}}( c\triangleright_{z_{ik}}a).\]
 That is $(X, \triangleright_{z_{ij}})$ is a $p$-shelf.

If in addition for  all $a,b \in X,$ $h_a = h_b \Rightarrow a=b,$ then the equation above leads to 
\[ c \triangleright_{z_{ik}}(b \triangleright_{z_{jk}}a) =(c \triangleright_{z_{ij}} b)\triangleright_{z_{jk}}( c\triangleright_{z_{ik}}a).\]
 That is $(X, \triangleright_{z_{ij}})$ is a $p$-shelf.

We also assume that  $c\triangleright_{z_{ij}} a = c\triangleright_{z_{ij}} b$, then $q^{z_{ij}}_c h_{c\triangleright_{z_{ij}} a} = q^{z_{ij}}_c h_{c\triangleright_{z_{ij}} b}$, 
 by the fourth relation in \ref{qualg}, 
 we get $h_aq^{z_{ij}}_c = h_bq^{z_{ij}}_c$ and by the 
 invertibility of $q^{z_{ij}}_c$, $h_a = h_b,$ hence $a=b$. 
 \end{proof}
 \begin{rem}
 In Proposition \ref{proro} if $(X, \triangleright_{z_{ij}})$ is a finite magma, or such that $a\triangleright_{z_{ij}}$ is surjective for every $a\in X,$ $z_{i,j} \in Y$ and for all $a,b \in X,$ $h_a =h_b \Rightarrow h_a=h_b,$  then for all $a\in X,$ $z_{i,j} \in Y$, $a\triangleright_{z_{ij}}$ is bijective and hence $(X, \triangleright_{z_{ij}})$ is a $p$-rack. 
 \end{rem}
\begin{defn} \label{qualgd} ($p$-rack algebra.) A $p$-shelf algebra ${\cal Q}$ is a $p$-rack algebra if $(X, \triangleright_{z_{ij}}),$ $z_{i,j}\in Y$ is a $p$-rack.
\end{defn}
\begin{lemma} \label{lemmac}
Let ${\mathrm C}= \sum_{a\in X} h_a,$ then ${\mathrm C}$ is a central element of the $p$-rack algebra ${\cal Q}$. Also, $h_a^2 ={\mathrm C} h_a,$ for all $a \in X.$
\end{lemma}
\begin{proof} The proof is straightforward by means of the definition of 
the algebra ${\cal Q}$ and the fact that $a \triangleright_{z_{ij}}$ is bijective. By rescaling the elements $h_a,$ i.e. by setting $h_a' := C^{-1}h_a,$ we observe that all algebra relations (\ref{qualg}) hold and also $\sum_{a\in X} h'_a = 1_{{\cal Q}}.$ Also, it immediately follows that for all $a \in X,$ $h_a^2 = {\mathrm C} h_a,$  and after rescaling $h_a^{'2} = h_a'.$ 
\end{proof}
Henceforth, we consider without loss of generality, $\sum_{a\in X} h_a=1_{ {\cal Q}}$ and $h_a^2 = h_a$ for all $a\in X.$

\begin{lemma}\label{lemmac2} (Commuting quantities.)
Let $t^{z_{ik}} := \sum_{a\in X} q_a^{z_{ik}},$ for $z_{i,k} \in Y.$ Then,
$~t^{z_{jk}} t^{z_{ik}} = t^{z_{ik}} t^{z_{jk}},  $ for all $z_{i,j,k} \in Y.$
 \end{lemma}
\begin{proof}
The proof is straightforward by direct computation  by using of the relation $ q_a^{z_{jk}} q^{z_{ik}}_b = q^{z_{ik}}_b  
q^{z_{jk}}_{b \triangleright_{z_{ij}} a}$ and the fact that 
$ a\triangleright_{z_{ij}}$ is bijective.
\end{proof}

Having defined the $p$-rack algebra we may now show that there exists 
an associated universal ${\cal R}$-matrix, solution of the parametric Yang-Baxter equation.
\begin{pro} 
Let ${\cal Q}$ be a $p$-rack algebra and ${\cal R}^{z_{ij}} \in {\cal Q} \otimes {\cal Q}$ 
be an invertible element, such that ${\cal R}^{z_{ij}} = \sum_{a\in X} h_a \otimes q^{z_{ij}}_a,$ $z_{i,j}\in Y.$
Then ${\cal R}^{z_{ij}}$ satisfies the parametric Yang-Baxter equation,
\begin{equation}
{\cal R}^{z_{12}}_{12} {\cal R}^{z_{13}}_{13} {\cal R}^{z_{23}}_{23} = {\cal R}^{z_{23}}_{23} {\cal R}^{z_{13}}_{13} {\cal R}^{z_{12}}_{12}, \label{YBEu} \nonumber
\end{equation}
where ${\cal R}^{z_{12}}_{12} = \sum_{a\in X} h_a \otimes q^{z_{12}}_a \otimes 1_{\cal Q}, $ ${\cal R}^{z_{13}}_{13} = \sum_{a\in X} h_a  
\otimes 1_{\cal Q} \otimes q^{z_{13}}_a,$ and  ${\cal R}^{z_{23}}_{23} = \sum_{a\in X} 1_{\cal Q} \otimes h_a \otimes q^{z_{23}}_a.$
The inverse ${\cal R}$-matrix is $({\cal R}^{z_{ij}})^{-1} = \sum_{a\in X} h_a\otimes (q_a^{z_{ij}})^{-1}.$
\end{pro}
\begin{proof}
The proof is  a direct computation of the two sides of the Yang-Baxter equation 
(and use of the fundamental relations (\ref{qualg})):
The LHS of the Yang-Baxter equation becomes
\begin{eqnarray}
   \sum_{a,b,c \in X} h_a h_b \otimes  q^{z_{12}}_a  h_c \otimes q^{z_{13}}_b q^{z_{23 }}_c  =  
\sum_{a,b,c \in X} h_a \otimes  q^{z_{12}}_a  h_c \otimes q^{z_{13}}_a q^{z_{23}}_c \nonumber
\end{eqnarray}
whereas the RHS gives
\begin{eqnarray}
\sum_{a,b,c\in X}   h_b h_a \otimes  h_c  q^{z_{12}}_a \otimes q^{z_{23}}_c q^{z_{13}}_b=  
\sum_{a,b,c\in X}    h_a \otimes   q^{z_{12}}_a h_{a\triangleright_{z_{12}}c} \otimes q^{z_{13}}_a q^{z_{23}}_{a \triangleright_{z_{12}}c}, \nonumber
\end{eqnarray}
by setting  ${a \triangleright_{z_{12}} c} = \hat c$ in the final expression for the RHS (using that $a \triangleright_{z_{12}}$ is bijecive) we show that LHS$=$RHS.

We recall from Lemma \ref{lemmac} that $\sum_{a\in X}{h_a} = 1_{\cal Q},$ then $({\cal R}^{z_{ij}})^{-1} = \sum_{a\in X} h_a\otimes (q_a^{z_{ij}})^{-1}.$
\end{proof}

\begin{rem} \label{remfu} 
{\bf Fundamental representation:} We first recall the $n\times n$ matrices $e_{i,j},$ with elements $(e_{i,j})_{k,l} = \delta_{i,k} \delta_{j, l}.$ Let ${\cal Q}$ be a $p$-rack algebra and $\rho: {\cal Q} \to \EEnd({\mathbb C}^n)$ ($V$ is an $n$-dimensional vector space) be the map defined by
\begin{equation}
q^{z_{ij}}_a \mapsto \sum_{x \in X} e_{x, a \triangleright_{z_{ij}} x}, \quad h_a\mapsto e_{a,a}. \label{remfu1}
\end{equation}
This is a representation of the rack algebra if and only of $(X, \triangleright_{z_{ij}})$ is a $p$-rack for all $z_{i,j} \in Y.$
Then ${\cal R}^{z_{ij}}\mapsto R^{z_{ij}}= 
\sum_{a,b\in X} e_{b,b} \otimes e_{a, b\triangleright_{z_{ij}}a},$ which is the linearized version of the $p$-rack solution.{We also note that $(R^{z_{ij}})^{-1} =(R^{z_{ij}})^{T},$ where $^T$ denotes total transposition.}
\end{rem}
We recall that $a\triangleright_{z_{ij}}: X \to X$ is a bijection, 
and ($R^{z_{12}})^{-1} = \sum_{a,b\in X} e_{b,b} \otimes e_{b\triangleright_{z_{12}} a, a}.$

We briefly formulate the Faddeev-Reshetikhin-Takhtajan (FRT) construction \cite{FRT} in order to check the consistency of our algebraic construction. From the parametric Yang-Baxter equation (\ref{YBEu}, and after recalling the representations of Remark
 \ref{remfu} and setting:
$(\rho \otimes \id){\cal R}^{z_{ij}} := L^{z_{ij}} = \sum_{a\in X}e_{a,a} \otimes q^{z_{ij}}_a,$ 
$~(\id \otimes \rho){\cal R}^{z_{ij}}= \hat L^{z_{ij}}= \sum_{a,b \in X} h_b \otimes e_{a, b \triangleright_{z_{ij}}a},$ and \\ 
$(\rho \otimes \rho){\cal R}^{z_{ij}}:=R^{z_{ij}}= \sum_{a,b  \in X} e_{b, b} \otimes e_{a, b \triangleright_{z_{ij}}a},$ 
we derive consistent algebraic relations:
\begin{equation}
    R^{z_{12}}_{12} L^{z_{13}}_{13} L^{z_{23}}_{23} = L^{z_{23}}_{23} L^{z_{13}}_{13} R^{z_{12}}_{12},  ~~\hat L^{z_{12}}_{12} \hat L^{z_{13}}_{13} R^{z_{23}}_{23} = R^{z_{23}}_{23}  \hat L^{z_{13}}_{13} \hat L^{z_{12}}_{12},
~~L^{z_{12}}_{12} R^{z_{13}}_{13}
\hat L^{z_{23}}_{23} = \hat L^{z_{23}}_{23} R^{z_{13}}_{13} L ^{z_{12}}_{12}, \label{a}
\end{equation}
which lead to the $p$-rack algebra given in Definition (\ref{qualg}) and provide a strong
consistency check on the algebraic relations (\ref{qualg}).

\begin{pro} \label{basic01} (Algebra homomorphism)
Let ${\cal Q}$ 
be a $p$-rack algebra and 
${\cal R}^{z_{ij}}= \sum_{a\in X} h_a \otimes q^{z_{ij}}_a \in {\cal Q} \otimes {\cal Q}$ 
be a solution of the Yang-Baxter equation.
We also define for $z_{i,j,k} \in Y,$ $\Delta'_{z_{ij}}: {\cal Q} \to {\cal Q} \otimes {\cal Q},$ such that for all $a \in X,$ 
\[\Delta'_{z_{jk}}((q^{z_{ik}}_a)^{\pm 1}) := (q^{z_{ij}}_a)^{\pm 1} \otimes (q^{z_{ik}}_a)^{\pm 1}, \quad \Delta'_{z_{ij}}(h_a) := \sum_{b,c \in X} 
h_b \otimes h_c\Big |_{b\triangleright_{z_{ij}} c = a}.\]
Then the following statements hold: 
\begin{enumerate}
\item $\Delta'_{z_{ij}}$ is a ${\cal Q}$ algebra homomorphism for all $z_{i,j} \in Y.$
\item ${\cal R}^{z_{jk}} \Delta'_{z_{jk}}(q_a^{z_{ik}}) = \Delta^{'(op)}_{z_{kj}}(q_a^{z_{ik}}) {\cal R}^{z_{jk}},$ 
for all $a \in X,$ $z_{i,j,k} \in Y.$ Recall, $\Delta^{'(op)}_{z_{ij}} := \pi \circ \Delta'_{z_{ij}},$ where $\pi$ is the flip map.
\end{enumerate}
\end{pro}
\begin{proof}
For our proof we use the definition of the $p$-rack algebra.
\begin{enumerate}
\item In order to show that  $\Delta'_{z_{ij}}$ is an algebra homomorphism 
it suffices to show the following statements for all $z_{i,j,k,l} \in Y.$  
We first check that indeed:\\ $\Delta'_{z_{kl}}((q_a^{z_{jl}})^{-1})\Delta'_{z_{kl}}(q_a^{z_{jl}})=\Delta'_{z_{kl}}(q_a^{z_{jl}})\Delta'_{z_{kl}}((q_a^{z_{jl}})^{-1}) = 1_{\cal Q} \otimes 1_{\cal Q}.$
Moreover, 
 \begin{eqnarray}
    && \Delta'_{z_{kl}}(q_a^{z_{jl}})\Delta'_{z_{kl}}(q_b^{z_{il}}) = (q_a^{z_{jk}} \otimes q_b^{z_{jl}}) (q_b^{z_{ik}} \otimes q_b^{z_{il}}) = \nonumber\\ && (q_b^{z_{ik}} \otimes q_b^{z_{il}}) (q_{b\triangleright_{z_{ij}} a}^{z_{jk}} \otimes q_{b\triangleright_{z_{ij}} a}^{z_{jl}}) =\Delta'_{z_{kl}}(q_b^{z_{il}})\Delta'_{z_{kl}}(q_{b \triangleright_{z_{ij}}a}^{z_{jl}}). \nonumber \end{eqnarray}
Via the algebraic relations (\ref{qualg}) we obtain
\begin{eqnarray}
&& \Delta'_{z_{jk}}(h_a) \Delta'_{z_{jk}}(q_b^{z_{ik}}) = (h_{a_1} \otimes h_{a_2}\Big|_{a_1 \triangleright_{z_{jk}} a_2 =a })(q_{b}^{z_{ij}} \otimes q_b^{z_{ik}}) = \nonumber\\
&& (q_{b}^{z_{ij}} \otimes q_b^{z_{ik}}) (h_{b\triangleright_{z_{ij}} a_1} \otimes h_{b\triangleright_{z_{ik}} a_2}\Big|_{b\triangleright_{z_{ik}}(a_1 \triangleright_{z_{jk}} a_2) ={b\triangleright_{z_{ik}} a}}) = \Delta'_{z_{jk}}(q_b^{z_{ik}}) 
\Delta'_{z_{jk}}(h_{b\triangleright_{z_{ik}}a}), \nonumber
\end{eqnarray}
where we have used the self-distributivity for racks, $(b\triangleright_{z_{ij}} a_1) \triangleright_{z_{jk}}(b\triangleright_{z_{ik}} a_2) = b \triangleright_{z_{ik}} (a_1 \triangleright_{z_{jk}} a_{2}).$ Similarly, it is straightforward to show via (\ref{qualg}) that
$\Delta'_{z_{jk}}(h_a)\Delta'_{z_{jk}}(h_b) =\delta_{a,b} \Delta'_{z_{jk}}(h_a).$ And this concludes the  proof of the first part.
\item  For the proof of the second part we first notice that $\Delta_{z_{kj}}^{'(op)}(q_b^{z_{ij}}) = \Delta'_{z_{jk}}(q_b^{z_{ij}}).$ 
We then compute for all 
$b \in X,$ $z_{i,j,k} \in Y$:
\begin{eqnarray}
    && {\cal R}^{z_{jk}}\Delta'_{z_{jk}}(q_b^{z_{ik}})= (\sum_{a\in X} h_a \otimes q_a^{z_{jk}}) (q_b^{z_{ij}} \otimes q_b^{z_{ik}}) = \nonumber\\
    && (q_b^{z_{ij}} \otimes q_b^{z_{ik}}) (\sum_{a\in X} h_{b \triangleright_{z_{ij}} a} \otimes q_{b \triangleright_{z_{ij} a}}^{z_{jk}}) = \Delta'_{z_{jk}}(q_b^{z_{ik}}) {\cal R}^{z_{jk}}. \nonumber  
         \end{eqnarray}
         \qedhere 
         \end{enumerate}
\end{proof}

\begin{lemma}
 (Parametric coassociativity.)
Let ${\cal Q}$ 
be a $p$-rack algebra.
We also define for $z_{i,1,2, \ldots, n} \in Y,$ 
$\Delta^{'(n)}_{z_{1 2 \ldots n}}: {\cal Q}  
\to {\cal Q}^{\otimes n},$ such that
\begin{eqnarray}
    && \Delta^{'(n)}_{z_{12\ldots n}}((q^{z_{in}}_a)^{\pm 1}) := 
    (q^{z_{i1}}_a)^{\pm 1} \otimes (q^{z_{i2}}_a)^{\pm 1}\otimes \ldots \otimes (q_a^{z_{in}})^{\pm 1} \label{con} \\ 
    && \Delta^{'(n)}_{z_{12 \ldots n}}(h_a) := 
    \sum_{a_1, \ldots,  a_n \in X} 
h_{a_1} \otimes h_{a_2} \otimes \ldots \otimes h_{a_n}\Big |_{a_1\triangleright_{z_{1n}}( a_2\triangleright_{z_{2n}} ( \ldots (a_{n-1} \triangleright_{z_{n-1 n} a_n})\ldots ))= a}. \label{con2b}
\end{eqnarray}
Then,
\begin{enumerate}
\item For all $a \in X,$ $z_{i, 1,2, \ldots n} \in Y,$
\begin{eqnarray} 
&&\Delta_{z_{12...n}}^{'(n)}((q_a^{z_{in}})^{\pm 1}):= (\Delta^{'(n-1)}_{z_{12...n-1}} \otimes \id)
\Delta'_{z_{n-1n}}((q_a^{z_{in}})^{\pm 1}) = (\id \otimes \Delta^{'(n-1)}_{z_{23...n}}) 
\Delta'_{z_{1n}}((q_a^{z_{in}})^{\pm 1}) \label{a2b}\\
&& \Delta_{z_{12...n}}^{'(n)}(h_a):=  (\id \otimes \Delta^{'(n-1)}_{z_{23...n}}) \Delta'_{z_{1n}}(h_a). \label{b2b}
\end{eqnarray}

\item For all $a, b \in X,$ $z_{i, 1,2, \ldots n} \in Y,$ 
$\Delta^{'(n)}_{z_{12\ldots n}}$ is an algebra homomorphism. 
\hfill \qedhere
\end{enumerate}
\end{lemma}
\begin{proof}
 These are shown by direct computation and use of the $p$ self-distributivity. \end{proof}

\begin{itemize}
\item Graphical representations of $\Delta'_{z_{12}},$ and the parametric 
coassociativity via binary trees:
\begin{enumerate}
\item We graphically depict below the parametric co-product $\Delta'_{z_{12}}:$

$ $

\begin{center}
\begin{tikzpicture}[
roundnode/.style={circle, draw=black!100, fill=gray!30,  thick, minimum size=7mm},
]
\node[roundnode]        (uppercircle) {$\Delta_{z_{12}}$};
\end{tikzpicture}
  $\qquad = \qquad $
\begin{tikzpicture}[thick,level/.style={sibling distance=30mm/#1}]
\node [vertex] {$z_{12}$}
    child { node {1}
     }
     child { node {2}
%
  };
  \end{tikzpicture}

\end{center}

\begin{center}
Figure 1.
\end{center}

$ $

\item We also see below the typical binary tree graphical representation of the 
parametric co-associativity condition for $n=3$ and $\Delta^{'(3)}_{z_{123}},$ i.e  
$\Delta^{'(3)}_{z_{123}}: = (\Delta'_{z_{12}} \otimes \id)\Delta'_{z_{13}} = 
(\id \otimes \Delta'_{z_{23}}) \Delta'_{z_{12}}:$

$ $

\begin{center}
\begin{tikzpicture}[
roundnode/.style={circle, draw=black!100, fill=gray!30, thick, minimum size=7mm},
]
\node[roundnode]        (uppercircle) {$\Delta^{(3)}_{z_{123}}$};
\end{tikzpicture}
  $\qquad = \qquad $
  \begin{tikzpicture}[thick,level/.style={sibling distance=30mm/#1}]
\node [vertex] {$z_{23}$}
  child {
   node [vertex] (a) {$z_{12}$}
    child { node {1}
     }
     child { node {2}
    }
  }
  child { node {3}
%
  };
\end{tikzpicture}
=
\begin{tikzpicture}[thick,level/.style={sibling distance=30mm/#1}]
\node [vertex] {$z_{13}$}
  child { node {1}
    }
  child { node{2}
    node [vertex] {$z_{23}$}
    child { node {2}
%
    }
    child { node {3}    }
   } ;
\end{tikzpicture}
\end{center}

\begin{center}
Figure 2.
\end{center}

$ $

\noindent Notice that $\Delta'_{z_{12}}(h_a)$ is still graphically represented by Figure 1, however the parametric coassociativity does not hold in this case and the three coproduct  $\Delta_{z_{123}}^{'(3)}$ is then represented by the right part of the graphical equation Figure 2. 

$ $

\item In general, in the case where coassociativity holds the $n^{th}$ coproduct $\Delta_{z_{12\ldots n}}^{'(n)}(q_a^{z_{kn}}),$ $a\in X,$ $z_{k,1,2, \ldots n} \in Y$ is depicted by $2^{n-2}$ equivalent diagrams:

$ $

\begin{center}
\begin{tikzpicture}[
roundnode/.style={circle, draw=black!100, fill=gray!30, thick, minimum size=7mm},
]
\node[roundnode]        (uppercircle) {$\Delta^{(n)}_{z_{12\ldots n}}$};
\end{tikzpicture}
$\qquad = \qquad $
  \begin{tikzpicture}[thick,level/.style={sibling distance=30mm/#1}]
\node [vertex] {$z_{n-1n}$}
  child {
   node [vertex] (a) {$\Delta^{(n-1)}_{z_{12 \ldots n-1}}$}
  }
  child { node {n}
%
  };
\end{tikzpicture}
=
\begin{tikzpicture}[thick,level/.style={sibling distance=30mm/#1}]
\node [vertex] {$z_{1n}$}
  child { node {1}
    }
  child { node{2}
    node [vertex] {$\Delta^{(n-1)}_{z_{23\ldots n}}$}
%
   };
\end{tikzpicture}
\end{center}
\begin{center}
Figure 3.
\end{center}

 $ $
 
 Unfolding $\Delta^{'(n-1)}$ in the LHS and RHS of Figure 3 yields $2^{n-2}$ binary tree diagrams. $\Delta^{'(n)}_{z_{12 \ldots n}}(h_a)$ on the other hand (no co associativity applies) is depicted by only one diagram, shown in the right part of Figures 2 $\&$ 3. 
 \end{enumerate}
\end{itemize}
\begin{itemize}

\item  {\bf Example.} The four equivalent binary trees that depict $\Delta_{z_{1234}}^{'(4)}(q_a^{z_{k4}}),$ emerging from Figures 2 and 3:
\begin{center}
\begin{tikzpicture}[thick,level/.style={sibling distance=40mm/#1}]
\node [vertex] {$z_{34}$}
  child {
    node [vertex] (a) {$z_{23}$}
    child {
      node [vertex] {$z_{12}$}
      child {
        node {$1$}
      }
      child {node {$2$}}
    }
    child {
      node  {$3$}
    }
  }
  child {
    node{$4$}
  };
\end{tikzpicture}
\begin{tikzpicture}[thick,level/.style={sibling distance=40mm/#1}]
\node [vertex] {$z_{14}$}
  child {
    node  {$1$}
    }
  child {
    node [vertex] {$z_{24}$}
    child {
      node  {$2$}
    }
    child {
      node [vertex] {$z_{34}$}
      child {node  {$3$}}
      child {node {$4$}}
    }
  };
\end{tikzpicture}

\begin{tikzpicture}[thick,level/.style={sibling distance=40mm/#1}]
\node [vertex] {$z_{34}$}
  child {
    node [vertex] (a) {$z_{13}$}
    child {
      node {$1$}
    }
    child {
      node [vertex] {$z_{23}$}
      child {node  {$2$}}
      child {node  {$3$}}
    }
  }
  child {
    node  {$4$}
  };
\end{tikzpicture}
\begin{tikzpicture}[thick,level/.style={sibling distance=40mm/#1}]
\node [vertex] {$z_{14}$}
  child {
    node  (a) {$1$}
  }
  child {
    node [vertex] {$z_{34}$}
    child {
      node [vertex] {$z_{23}$}
      child {node  {$2$}}
      child {node  {$3$}}    }
    child {
      node  {$4$}
      }
  };
\end{tikzpicture}
\end{center}
\begin{center}
Figure 4.
\end{center}
$\Delta_{z_{1234}}^{'(4)}(h_a)$ is depicted by the top right binary tree above.
\end{itemize}

\begin{defn}  \label{qualgd2} 
(The restricted $p$-rack algebra.) A $p$-rack algebra ${\cal Q}$ 
is called a restricted $p$-rack algebra if  for all $z_{i{,}j} \in Y$ there exits 
a binary operation $\bullet_{z_{ij}}: X\times X \to X,$  $(a,b) \mapsto a \bullet_{z_{ij}} b,$ such that,
$a\bullet_{z_{ij}},$ is
bijective and  $a\bullet_{z_{ji}}b = b \bullet_{z_{ij}}(b \triangleright_{z_{ij}} a),$ for all $a,b \in X,$ $z_{i,j}\in Y.$
\end{defn}
Recall Example \ref{bullet}, where the condition of Definition \ref{qualgd2} is satisfied.

\begin{thm} \label{basica1}  
Let ${\cal Q}$ 
be the restricted $p$-rack algebra  and 
${\cal R}^{z_{ij}}= \sum_{a\in X} h_a \otimes q^{z_{ij}}_a \in {\cal Q} \otimes {\cal Q}$ 
be a solution of the Yang-Baxter equation.
Moreover, assume that for all $z_{i,j,k} \in Y,$ 
$a,b\in X,$  $(b\triangleright_{z_{ij}} a_1) \bullet_{z_{jk}}(b\triangleright_{z_{ik}} a_2) = b \triangleright_{z_{ik}}(a_1 \bullet_{z_{jk}}a_{2}).$  We also define for $z_{i,j,k} \in Y,$ $\Delta_{z_{ij}}: {\cal Q} \to {\cal Q} \otimes {\cal Q},$ such that for all $a \in X,$ 
\[\Delta_{z_{jk}}((q^{z_{ik}}_a)^{\pm 1}) := (q^{z_{ij}}_a)^{\pm 1} \otimes (q^{z_{ik}}_a)^{\pm 1}, \quad \Delta_{z_{ij}}(h_a) := \sum_{b,c \in X} 
h_b \otimes h_c\Big |_{b\bullet_{z_{ij}} c = a}.\]
Then the following statements hold: 
\begin{enumerate}
\item $\Delta_{z_{ij}}$ is a ${\cal Q}$ algebra homomorphism for all $z_{i,j} \in Y.$
\item ${\cal R}^{z_{jk}} \Delta_{z_{jk}}(y) = \Delta^{(op)}_{z_{kj}}(y) {\cal R}^{z_{jk}},$ 
for all $z_{j,k} \in Y,$ $~y \in {\cal Q}.$ Recall $\Delta^{(op)}_{z_{ij}} := \pi \circ \Delta_{z_{ij}},$
where $\pi$ is the flip map.
\end{enumerate}
\end{thm}
\begin{proof}
For the proof we use the definition of the restricted the $p$-rack algebra.
\begin{enumerate}
\item In order to show that  $\Delta_{z_{ij}}$ is an algebra homomorphism it 
suffices to show the following statements for all $z_{i,j,k,l} \in Y.$  
We note that $\Delta_{z_{kl}}(q_a^{z_{jl}})\Delta_{z_{kl}}((q_a^{z_{jl}})^{-1}) = 1_{\cal Q} \otimes 1_{\cal Q}$
and $\Delta_{z_{kl}}(q_a^{z_{jl}})\Delta_{z_{kl}}(q_b^{z_{il}})  =
    \Delta_{z_{kl}}(q_b^{z_{il}})\Delta_{z_{kl}}(q_{b \triangleright_{z_{ij}}a}^{z_{jl}})$ are shown as in the first part of the proof of Proposition \ref{basic01}.
Via the algebraic relations (\ref{qualg}) we obtain
\begin{eqnarray}
&& \Delta_{z_{jk}}(h_a) \Delta_{z_{jk}}(q_b^{z_{ik}}) = (h_{a_1} \otimes h_{a_2}\Big|_{a_1 \bullet_{z_{jk}} a_2 =a })(q_{b}^{z_{ij}} \otimes q_b^{z_{ik}}) = \nonumber\\
&& (q_{b}^{z_{ij}} \otimes q_b^{z_{ik}}) (h_{b\triangleright_{z_{ij}} a_1} \otimes h_{b\triangleright_{z_{ik}} a_2}\Big|_{a_1 \bullet_{z_{jk}} a_2 =a}) = \Delta_{z_{jk}}(q_b^{z_{ik}}) 
\Delta_{z_{jk}}(h_{b\triangleright_{z_{ik}}a}), \nonumber
\end{eqnarray}
where we have used $(b\triangleright_{z_{ij}} a_1) \bullet_{z_{jk}}(b\triangleright_{z_{ik}} a_2) = b \triangleright_{z_{ik}} (a_1 \bullet_{z_{jk}} a_{2}).$ Similarly, it is straightforward to show via (\ref{qualg}) that
$\Delta_{z_{jk}}(h_a)\Delta_{z_{jk}}(h_b) =\delta_{a,b} \Delta_{z_{jk}}(h_a).$ And this concludes the  proof of the first part.
\item  We recall the proof of the second part of the proof of Proposition \ref{basic01}. Also,
    \begin{eqnarray}
   && \Delta_{z_{ji}}^{(op)}(h_a){\cal R}^{z_{ij}} = ( \sum_{a_1, a_2 \in X }h_{a_2} \otimes h_{a_1}\Big |_{a_{1} \bullet_{z_{ji}} a_2=a }) (\sum_{b\in X} h_b \otimes q_{b}^{z_{ij}}) = \nonumber\\
   && (\sum_{b\in X} h_b \otimes q_{b}^{z_{ij}}) ( \sum_{a_1, a_2 \in X } h_{a_2} \otimes h _{a_2 \triangleright_{z_{ij} a_1}}\Big |_{a_{2}\bullet_{z_{ij}} (a_2 \triangleright_{z_{ij}}a_1) =a})    =  {\cal R}^{z_{ij}}\Delta_{z_{ij}}(h_a), \end{eqnarray}
   where we have used $a_1 \bullet_{z_{ji}} a_2 = a_{2}\bullet_{z_{ij}} (a_2 \triangleright_{z_{ij}}a_1).$
   \hfil \qedhere
\end{enumerate}
\end{proof}


\begin{pro} \label{basicb2} 
(Parametric (co)-associativity.) 
Let ${\cal Q}$ 
be the restricted $p$-rack algebra, assume also that for all 
$a, b,c\in X$ and $z_{i, j,k} \in Y,$  $(b\triangleright_{z_{ij}} a) 
\bullet_{z_{jk}}(b\triangleright_{z_{ik}} c) = 
b\triangleright_{z_{i k}}(a \bullet_{z_{jk}} c)$ and  $(a \bullet_{z_{ij}} b) 
\bullet_{z_{j k}} c = a \bullet_{z_{i k}}(b \bullet_{z_{jk}} c).$ 

We also define for $z_{i,1,2, \ldots, n} \in Y,$ 
$\Delta^{(n)}_{z_{1 2 \ldots n}}: {\cal Q}  
\to {\cal Q}^{\otimes n},$ such that
\begin{eqnarray}
    && \Delta^{(n)}_{z_{12\ldots n}}((q_a^{z_{in}})^{\pm 1})= 
    (q^{z_{i1}}_a)^{\pm 1} \otimes (q^{z_{i2}}_a)^{\pm 1}\otimes \ldots \otimes (q_a^{z_{in}})^{\pm 1}, \label{conb} \\ 
    && \Delta^{(n)}_{z_{12 \ldots n}}(h_a) := 
    \sum_{a_1, \ldots,  a_n \in X} 
h_{a_1} \otimes h_{a_2} \otimes \ldots \otimes h_{a_n}\Big |_{\Pi_{z_{1\ldots n}}(a_1, a_2, \ldots, a_n)= a}, \label{con2}
\end{eqnarray}
where for all $a_1, a_2, \ldots, a_n \in X,$ $z_1, \ldots, z_n \in Y:$ 
\begin{eqnarray}
\Pi_{z_{12}}(a_1, a_2): &=& a_1 \bullet_{z_{12}} a_2 \label{prod} \\
\Pi_{z_{12\ldots n}}(a_1, a_2, \ldots, a_n): &=& a_1 \bullet_{z_{1 n}}(a_2 \bullet_{z_{2 n}} (a_3 \ldots \bullet_{z_{n-2  n}}( a_{n-1} \bullet_{z_{n-1 n}} a_n )\ldots )) \nonumber\\
&=& ((\ldots ((a_1 \bullet_{z_{12}} a_2 ) \bullet_{z_{ 2 3}} a_3) \ldots )\bullet_{z_{n-2 n-1}} a_{n-1}) \bullet_{z_{ n-1 n}} a_n, ~~~n>2. \nonumber
\end{eqnarray}
Then: 
\begin{enumerate}
\item For all $a \in X,$ $z_{i, 1,2, \ldots n} \in Y,$ the ``parametric'' coassociatvity holds, $y \in \{q_a^{z_{in}}, h_a\}$:
\begin{eqnarray} 
&&\Delta_{z_{12...n}}^{(n)}(y)  := (\Delta^{(n-1)}_{z_{12...n-1}} \otimes \id)\Delta_{z_{n-1n}}(y) = (\id \otimes \Delta^{(n-1)}_{z_{23...n}}) \Delta_{z_{1n}}(y). \label{b2}
\end{eqnarray}

\item For all $a, b \in X,$ $z_{i, 1,2, \ldots n} \in Y,$ 
$\Delta^{(n)}_{z_{12\ldots n}}$ is an algebra homomorphism.
\end{enumerate}
\end{pro}
\begin{proof}
The proof of both parts is straightforward.
\begin{enumerate}
\item  The $n$-coassociativity (\ref{b2}) for the coproducts of $q_a^{z_{ij}}$ for all $a\in X,$ $z_{i,j} \in Y$ is 
immediately shown by direct computation.
Also, using (\ref{prod}) we conclude for 
$n>2:$ \[\Pi_{z_{12\ldots n}}(a_1, a_2, \ldots, a_n) = a_1 \bullet_{z_{1 n}} \Pi_{z_{2\ldots n}}(a_2, \ldots, a_n) = \Pi_{z_{12\ldots n-1}}(a_1, a_2, \ldots, a_{n-1}) \bullet_{z_{n-1n}} a_n.\] 
Then (\ref{b2}) immediately follows for the $n$-coproduct of $h_a,$ $a \in X$.

\item Indeed, the $p$-rack algebra relations hold for the $n$-coproducts 
of the generators (see also Theorem \ref{basica1}). The algebra relations are shown 
as a generalization of the proof of part 1 of 
Theorem \ref{basica1} bearing also in mind expressions (\ref{con}), (\ref{con2}).
\hfill \qedhere
\end{enumerate}
\end{proof}

The $n$-coproducts as defined in Proposition \ref{basicb2} are naturally depicted by binary trees (see Figures 1-4, pages 17-18).
We provide in the following corollary a concrete statement of quantum integrability 
as we identify an explicit set of mutually commuting non-local quantities.

\begin{cor} \label{corc1} (Commuting non-local quantities.) 
We define for all $z_{i, k_1, \ldots, k_n} \in Y,$ 
\[{\mathfrak t}^{z_{ik_1 \ldots k_n}} := \sum_{a \in X} \Delta_{z_{k_1 k_2 \ldots k_n}}^{(n)}(q_a^{z_{ik_n}}), \]
then $~{\mathfrak t}^{z_{jk_1 \ldots k_n}} {\mathfrak t}^{z_{ik_1 \ldots k_n}} =
{\mathfrak t}^{z_{ik_1 \ldots k_n}}{\mathfrak t}^{z_{jk_1 \ldots k_n}},$ 
for all $z_{i,j, k_1, \ldots, k_n} \in Y.$
\end{cor}
\begin{proof}
This is a consequence of Lemma \ref{lemmac2} and the form of the coproduct 
$\Delta_{z_{k_1 k_2 \ldots k_n}}^{(n)}(q_a^{z_{ik_n}}).$
\end{proof}

\begin{lemma} \label{quasit}
    Let ${\cal R}^{z_{ij}}$ be a solution of the parametric Yang-Baxter equation and define:
    \begin{eqnarray}
    T^{z_{12\ldots  n+1}}_{1,23\ldots n+1}:= {\cal R}_{1n+1}^{z_{1n+1}}  {\cal R}_{1n}^{z_{1n}}\ldots  {\cal R}^{z_{12}}_{12}, \qquad  T_{12\ldots n,n+1}^{z_{12\ldots n+1}}:= {\cal R}_{1n+1}^{z_{1n+1}}{\cal R}_{2n+1}^{z_{2n+1}} \ldots {\cal R}_{nn+1}^{z_{nn+1}}. \label{TT}   \end{eqnarray}
Let also ${\cal Q}$ be the restricted $p$-rack algebra, 
${\cal R}^{z_{ij}} = \sum_{a \in X} h_a \otimes q_{a}^{z_{ij}} \in {\cal Q} \otimes  {\cal Q}$ and 
for all $z_{i, j,k} \in Y,$ $a, b, c \in X,$ 
$(b\triangleright_{z_{ij}} a) \bullet_{z_{jk}}(b\triangleright_{z_{ik}} c) = b\triangleright_{z_{i k}}(a \bullet_{z_{jk}}c),$ $q^{z_{jk}}_{a}q^{z_{ik}}_{b} = q^{z_{ik}}_{a\bullet_{z_{ji}} b}$ and $(a \bullet_{z_{ij}} b) \bullet_{z_{j k}} c = a \bullet_{z_{i  k}}(b \bullet_{z_{jk}} c),$  then
\begin{eqnarray}
&& T_{1,2\ldots n+1}^{z_{12 \ldots n+1}} =\sum_{a\in X} h_a \otimes \Delta^{(n)}_{z_{23\ldots n+1}}(q_a^{z_{1n+1}}) = (\id \otimes \Delta_{z_{23\ldots n+1}}^{(n)}){\cal R}^{z_{1n+1}} \label{coco}\\ 
&& T_{12\ldots n,n+1}^{z_{12\ldots n+1}} =\sum_{a \in X} \Delta_{z_{12\ldots n}}^{(n)}(h_a) \otimes q^{z_{n n+1}}_{a} = (\Delta^{(n)}_{z_{12\ldots n}} \otimes \id){\cal R}^{z_{nn+1}}. \label{coco2}
\end{eqnarray}
\end{lemma}
\begin{proof}
The proof is immediate:
\begin{eqnarray}
&& T_{1,2\ldots n+1}^{z_{12 \ldots n+1}} = 
\sum_{a\in X} h_a \otimes q_a^{z_{12}} \otimes q_a^{z_{13}} \otimes \ldots \otimes q_a^{z_{1n+1}},\nonumber\\ 
&& T_{12\ldots n,n+1}^{z_{12\ldots n+1}} = \sum_{a \in X} \sum_{a_1, \ldots, a_n \in X}  h_{a_1} \otimes \ldots \otimes h_{a_n}\Big|_{a:= \Pi_{z_{12\ldots n}}(a_1, a_2, \ldots, n)}\otimes q^{z_{ n n+1}}_{a}. \nonumber
\end{eqnarray}    
Recalling the definitions of the $n$-coproducts in Proposition \ref{basicb2}, we arrive at (\ref{coco}), (\ref{coco2})\end{proof}

Lemma \ref{quasit} with Theorem \ref{basica1} and Proposition \ref{basicb2} provide a structure that generalizes in a sense the notion of a quasi-triangular (quasi)-bialgebra \cite{Drinfel'd, Drinfel'd2} to the parametric frame. In the parameter free case the structure formulated in Lemma \ref{quasit}, Theorem \ref{basica1} and Proposition \ref{basicb2} corresponds indeed to a quasi-triangular Hopf algebra if $(X, \bullet)$ is group \cite{DoRySt}. Recall that quantities (\ref{TT}) are tensor realizations of the algebras defined by (\ref{a}) via the FRT construction \cite{FRT}.

\begin{rem} Theorem  \ref{basica1} and Proposition \ref{basicb2} can be generalized as follows. 
Let ${\cal Q}$ 
be the restricted $p$-rack algebra, assume also that for all 
$a, b,c\in X$ and $z_{i, j,k} \in Y,$ there exist $z_{\hat o}, z_{o} \in Y,$ such that
$(b\triangleright_{z_{ij}} a) 
\bullet_{z_{jk}}(b\triangleright_{z_{ik}} c) = 
b\triangleright_{z_{i \hat o}}(a \bullet_{z_{jk}} c)$ and  $(a \bullet_{z_{ij}} b) 
\bullet_{z_{ok}} c = a \bullet_{z_{i o}}(b \bullet_{z_{jk}} c).$ 

We also define for $z_{i,1,2, \ldots, n} \in Y,$ 
$\Delta^{(n)}_{z_{1 2 \ldots n}}: {\cal Q}  
\to {\cal Q}^{\otimes n},$ such that
\begin{eqnarray}
    && \Delta^{(n)}_{z_{12\ldots n}}((q_a^{z_{in}})^{\pm 1})= 
    (q^{z_{i1}}_a)^{\pm 1} \otimes (q^{z_{i2}}_a)^{\pm 1}\otimes \ldots (\otimes q_a^{z_{in}})^{\pm 1}, \label{conb2} \\ 
    && \Delta^{(n)}_{z_{12 \ldots n}}(h_a) := 
    \sum_{a_1, \ldots,  a_n \in X} 
h_{a_1} \otimes h_{a_2} \otimes \ldots \otimes h_{a_n}\Big |_{\Pi_{z_{1\ldots n}}(a_1, a_2, \ldots, a_n)= a}, \label{con22}
\end{eqnarray}
where for all $a_1, a_2, \ldots, a_n \in X,$ $z_1, \ldots, z_n \in Y:$ 
\begin{eqnarray}
\Pi_{z_{12}}(a_1, a_2): &=& a_1 \bullet_{z_{12}} a_2 \label{prod2} \\
\Pi_{z_{12\ldots n}}(a_1, a_2, \ldots, a_n): &=& a_1 \bullet_{z_{1 o}}(a_2 \bullet_{z_{2 o}} (a_3 \ldots \bullet_{z_{n-2  o}}( a_{n-1} \bullet_{z_{n-1 n}} a_n )\ldots )) \nonumber\\
&=& ((\ldots ((a_1 \bullet_{z_{12}} a_2 ) \bullet_{z_{ o 3}} a_3) \ldots {)}\bullet_{z_{o n-1}} a_{n-1}) \bullet_{z_{ o n}} a_n, ~~~n>2. \nonumber
\end{eqnarray}
Then, it is shown by direct computation:
\begin{enumerate}
\item For all $a \in X,$ $z_{i, 1,2, \ldots n} \in Y,$ a parametric coassociativity holds as
\begin{eqnarray} 
&&\Delta_{z_{12...n}}^{(n)}((q_a^{z_{in}})^{\pm 1}) := (\Delta^{(n-1)}_{z_{12...n-1}} \otimes \id)\Delta_{z_{n-1n}}((q_a^{z_{in}})^{\pm 1})= (\id \otimes \Delta^{(n-1)}_{z_{23...n}}) \Delta_{z_{1n}}((q_a^{z_{in}})^{\pm 1}). \nonumber\\
&&\Delta_{z_{12...n}}^{(n)}(h_a) := (\Delta^{(n-1)}_{z_{12...n-1}} \otimes \id)\Delta_{z_{on}}(h_a)= (\id \otimes \Delta^{(n-1)}_{z_{23...n}}) \Delta_{z_{1o}}(h_a) \label{b22}
\end{eqnarray}

\item For all $a, b \in X,$ $z_{i, 1,2, \ldots n} \in Y,$ 
$\Delta^{(n)}_{z_{12\ldots n}}$ is a ``weak'' algebra homomorphism
i.e. almost all the $p$-rack algebra relations hold, but  
\begin{eqnarray}
&& \Delta^{(n)}_{z_{12 \ldots n}}(h_a)\Delta^{(n)}_{z_{12 \ldots n}}(q_b^{z_{in}}) =\Delta^{(n)}_{z_{12 \ldots n}}(q_b^{z_{in}}) \Delta^{(n)}_{z_{12 \ldots n}}(h_{b \triangleright_{z_{i \hat o}} a}). \label{c4}
\end{eqnarray}
\end{enumerate}\end{rem}

\begin{exa}
    \label{exa2} Consider the binary operations $~\bullet_{z_{ij}},\  \triangleright_{z_{ij}}: X \times X \to X,$ such that $a\bullet_{z_{ij}} b = a \circ z_i+ b \circ z_j$ and $a\triangleright_{z_{ij}} b=- a \circ z_i \circ z_j^{-1} +b + a \circ z_{i}\circ z_{j}^{-1},$ where $(X,+,\circ)$ is a skew brace (see also Example \ref{bullet}), then one shows by direct computation that for all
$a ,b, c\in X,$ $z_{i,j,k} \in Y,$
\[(a\triangleright_{z_{ij}} b) \bullet_{z_{jk}}(a\triangleright_{z_{ik}} c) = a \triangleright_{z_{i \hat o}}(b \bullet_{z_{jk}}c ), \quad (a \bullet_{z_{ij}} b) \bullet_{z_{o k}} c = a \bullet_{z_{i o}}(b \bullet_{z_{jk}} c),\] where $z_o = z_{\hat o} =1$
and 
\[ \Pi_{z_1\ldots z_n}(a_1, a_2, \ldots, a_n)= a_1 \circ z_1 + 
a_2 \circ z_2  + \ldots +  a_{n-1} \circ z_{n-1} + a_n \circ z_{n}.  \]
\end{exa}

    \begin{exa} \label{exabasic}
Consider the fundamental representation of the restricted $p$-rack algebra ${\cal Q},$
$\rho: {\cal Q} \to \EEnd({\mathbb C}^n),$ $q^{z_{ij}}_a \mapsto {\mathrm q}^{z_{ij}}_a:=  \sum_{b\in X} e_{b, a\triangleright_{z_{ij}} b}, $ 
$h_a \mapsto e_{a,a},$ $a \in X$ $z_{i,j} \in Y,$ then:
\begin{eqnarray}
{\mathrm q}_a^{z_{jk}} {\mathrm q}_b^{z_{ik}} =  \sum_{c\in X} e_{c, b\triangleright_{z_{ik}}(a\triangleright_{z_{jk}} c)}.\label{conj0}
\end{eqnarray}
Recall also Example \ref{bullet} for e.g. $\xi =1$, i.e. let also $(X,+, \circ)$ be a skew brace and recall the binary operations for $z_{i,j} \in Y,$ $\bullet_{z_{ij}}, \triangleright_{z_{ij}}: X\times X \to X,$ $a\bullet_{z_{ij}} b = { a \circ z_i  + b\circ z_j} $ and  $a \triangleright_{z_{ij}} b = -a \circ z_i \circ z_j^{-1} + b + a\circ z_i \circ z_j^{-1},$ also 
\begin{equation}
a\bullet_{z_{ji}} b =  b \bullet_{z_{ij}} (b \triangleright_{z_{ij}} a) \quad \mbox{and} \quad
b \triangleright_{z_{ik}}(a\triangleright_{z_{jk}} c) = (a\bullet_{z_{ji}} b)\triangleright_{z_{ok}} c,
\end{equation}
where $z_o =1,$ then via (\ref{conj0}) we conclude that 
${\mathrm q}_a^{z_{jk}} {\mathrm q}_b^{z_{ik}} =  
{\mathrm q}^{z_{ok}}_{a\bullet_{z_{ji}} b} = {\mathrm q}_b^{z_{ik}}{\mathrm q}_{b\triangleright_{z_{ij}}a}^{z_{jk}}.$ 
\end{exa}

We conclude the subsection  by noting that
the rack (co)-homology has been studied in \cite{Andru, coho2, LebMan}.
One of the natural future questions is the generalization of the (co)-homological analysis 
in the case of parametric racks, based on the parametric co-structure 
constructed here. This issue however will be addressed separately in a future work.

\subsection{\texorpdfstring{$p$-set Yang-Baxter algebras}{}}
In this subsection we suitably extend the $p$-rack algebra ${\cal Q}$ in order to construct 
the algebras associated to general set-theoretic solutions of the parametric Yang-Baxter equation.

We start our analysis by defining the {\it decorated $p$-shelf algebra}.
\begin{defn}  
\label{setalgd0} (Decorated $p$-shelf algebra.) Let ${\cal Q}$ 
be a $p$-shelf algebra and
$\sigma^{z_{ij}}_a, \ \tau_a^{z_{ij}}: X\to X,$ $a \in X,$ $z_{i,j}\in Y.$
We say that the unital, associative algebra $\hat {\cal Q}$ over $k,$
generated by indeterminates $q^{z_{ij}}_a, (q^{z_{ij}}_a)^{-1}, h_a, \in {\cal Q}$ and 
$w^{z_{ij}}_a, (w^{z_{ij}}_a)^{-1} \in \hat {\cal Q},$ $a \in X,$ $1_{\hat {\cal Q}}= 1_{\cal Q}$ 
(the unit element)
and relations, for $a,b \in X,$ $z_{i,j,k} \in Y:$
\begin{eqnarray}
&& q_a^{z_{ij}} (q_{a}^{z_{ij}})^{-1} =(q_{a}^{z_{ij}})^{-1}q_a^{z_{ij}} = 
1_{\hat {\cal Q}},\quad q_a^{z_{jk}} q^{z_{ik}}_b = q^{z_{ik}}_b  
q^{z_{jk}}_{b \triangleright_{z_{ij}} a}, \quad  h_a  h_b =\delta_{a, b} h^2_a, 
\nonumber\\ && q^{z_{ij}}_b  h_{b\triangleright_{z_{ij}} a}= h_a q^{z_{ij}}_b, \quad
 w^{z_{ij}}_a (w^{z_{ij}}_a)^{-1} =1_{\hat {\cal Q}}, \quad  
w_a^{z_{ki}} w_b^{z_{ji}}= 
w^{z_{ji}}_{\sigma^{z_{jk}}_a(b)} w_{\tau^{z_{kj}}_{b}(a)}^{z_{ki}} \nonumber\\
&&
w_a^{z_{ji}} h_b = h_{\sigma^{z_{ij}}_a(b)} w_a^{z_{ji}}, 
\quad  w_a^{z_{kj}}q_b^{z_{ij}}= q_{\sigma_a^{z_{ik}}(b)}^{z_{ij}} w_a^{z_{kj}} 
 \label{qualgbb}
\end{eqnarray}
is a {\it decorated $p$-shelf algebra}.
\end{defn}

\begin{pro} 
\label{qua2} Let $\hat {\cal Q}$ be the decorated $p$-shelf algebra, and for all $a,b \in X,$ $h_a = h_b \Rightarrow a =b.$
Then for all $a,b,c \in X,$ $z_{i,j,k}\in Y:$
\begin{eqnarray}
&& \sigma^{z_{ik}}_a(\sigma^{z_{ij}}_b(c)) = \sigma^{z_{ij}}_{\sigma^{z_{jk}}_a\left(b\right)}(\sigma^{z_{ik}}_{\tau^{z_{jk}}_b\left(a\right)}(c))  \quad \& \quad   
\sigma^{z_{ik}}_c(b) \triangleright_{z_{ij}} \sigma^{z_{jk}}_{c}(a) = \sigma^{z_{jk}}_c(b \triangleright_{z_{ij}} a) \label{basicre}
\end{eqnarray}
and $\sigma^{z_{ij}}_a$ is injective.
\end{pro}
\begin{proof}
We compute $w^{z_{ik}}_a w^{z_{ij}}_b h_c$ using the associativity of the algebra,
\begin{eqnarray}
&& w^{z_{ki}}_a w^{z_{ji}}_b h_c =  w^{z_{ji}}_{\sigma^{z_{jk}}_a(b)} w^{z_{ki}}_{\tau^{z_{jk}}_b(a)} h_c =
h_{\sigma^{z_{ij}}_{\sigma^{z_{jk}}_{a}(b)}(\sigma^{z_{ik}}_{\tau^{z_{jk}}_{b}(a)}(c))}
w^{z_{ji}}_{\sigma^{z_{jk}}_b(c)} w^{z_{ki}}_{\tau^{z_{jk}}_{c}(b)},  \nonumber \label{wb1} \\
&& w^{z_{ki}}_a w^{z_{ji}}_b h_c =h_{\sigma^{z_{ik}}_a(\sigma^{z_{ij}}_b(c))}  w^{z_{ki}}_a w^{z_{ji}}_b  = h_{\sigma^{z_{ik}}_a(\sigma^{z_{ij}}_b(c))}
w^{z_{ji}}_{\sigma^{z_{jk}}_b(c)} w^{z_{ki}}_{\tau^{z_{jk}}_{c}(b)}. \label{wb2} \nonumber
\end{eqnarray}
From the equations above and the invertibility of $w^{z_{ij}}_a,$ for all $a\in X,$ $z_{i,j} \in Y,$ we conclude for all $a,b,c\in X,$ $z_{i,j,k}\in Y,$ 
\begin{equation}
h_{\sigma^{z_{ij}}_{\sigma^{z_{jk}}_{a}(b)}(\sigma^{z_{ik}}_{\tau^{z_{jk}}_{b}(a)}(c))}= h_{\sigma^{z_{ik}}_a(\sigma^{z_{ij}}_b(c))}\  
\Rightarrow \ \sigma^{z_{ij}}_{\sigma^{z_{jk}}_{a}(b)}(\sigma^{z_{ik}}_{\tau^{z_{jk}}_{b}(a)}(c))= \sigma^{z_{ik}}_a(\sigma^{z_{ij}}_b(c)),\    \label{basiko} \nonumber
\end{equation}
for all $a, b, c \in X,$ $z_{i,j,k} \in Y.$

We also compute $h_aq^{z_{ij}}_b w^{z_{kj}}_c$:
\begin{eqnarray}
&& h_aq^{z_{ij}}_b w^{z_{kj}}_c = h_aw^{z_{kj}}_cq^{z_{ij}}_{(\sigma^{z_{ik}}_c)^{-1}(b)} =   w^{z_{kj}}_c q^{z_{ij}}_{(\sigma^{z_{ik}}_c)^{-1}(b)}
h_{(\sigma^{z_{ik}}_c)^{-1}(b)\triangleright_{z_{ij}} (\sigma^{z_{jk}}_{c})^{-1}(a)} \nonumber\\
&&  h_aq^{z_{ij}}_b w^{z_{kj}}_c = q^{z_{ij}}_b h_{b \triangleright_{z_{ij}} a}w^{z_{kj}}_c = 
q^{z_{ij}}_b w^{z_{kj}}_c h_{(\sigma^{z_{ik}}_c)^{-1}(b \triangleright_{z_{ij}} a)} = w^{z_{kj}}_c q^{z_{ij}}_{(\sigma^{z_{ik}}_c)^{-1}(b)} h_{(\sigma^{z_{jk}}_c)^{-1}(b \triangleright_{z_{ij}} a)}. \nonumber
\end{eqnarray}
From the equations above and the invertibility of $q^{z_{ij}}_a,\ w^{z_{ij}}_a,$ for all $a\in X,$ $z_{i,j} \in Y,$ we conclude for all $a,b,c\in X,$ $z_{i,j,k}\in Y$:
\begin{equation}
h_{(\sigma^{z_{ik}}_c)^{-1}(b)\triangleright_{z_{ij}} (\sigma^{z_{jk}}_c)^{-1}(a)} = h_{(\sigma^{z_{jk}}_c)^{-1}(b \triangleright_{z_{ij}} a)}\ \Rightarrow \  
(\sigma^{z_{ik}}_c)^{-1}(b)\triangleright_{z_{ij}} (\sigma^{z_{jk}}_{c})^{-1}(a) = (\sigma^{z_{jk}}_c)^{-1}(b \triangleright_{z_{ij}} a).\nonumber
\end{equation}
From the latter it immediately follows, $ \sigma^{z_{ik}}_c(b) \triangleright_{z_{ij}} \sigma^{z_{jk}}_{c}(a) = \sigma^{z_{jk}}_c(b \triangleright_{z_{ij}} a),$ for all $a, b, c \in X,$ $z_{i,j,k} \in Y.$

We assume that $\sigma_a^{z_{ij}}(b) = \sigma_a^{z_{ij}}(c),$ then $h_{\sigma^{z_{ij}}_a(b)} w_a^{z_{ji}} =h_{\sigma^{z_{ij}}_a(c)} w_a^{z_{ji}},$ by the seventh equation in (\ref{qualgbb}), we obtain $w_a^{z_{ji}} h_b = w_a^{z_{ji}} h_c$ and by the invertibility of $w_a^{z_{ji}},$ we conclude that $h_b =h_c$ and hence $b=c.$
\end{proof}
It is interesting to note that the two key relations shown in Proposition \ref{qua2} are precisely the conditions  that appear in the definition of an admissible twist (see Definition \ref{def:twist:shelf}) and are intrinsic properties of the underlying  
associative algebra. That is, the  decorated $p$-shelf algebra guarantees the existence of an admissible twist and hence the existence of a generic invertible  set-theoretic solution.
\begin{defn}  
\label{setalgd} (Decorated $p$-rack algebra.) A decorated $p$-shelf algebra is a decorated $p$-rack algebra if $(X, \triangleright_{z_{ij}})$ is a $p$-rack and $\sigma_a^{z_{ij}}: X \to X$ is a bijection for all $a \in X,$ $z_{i,j}\in Y.$
\end{defn}
\begin{lemma} \label{lemmad}
Let ${\mathrm C}= \sum_{a\in X} h_a,$ then ${\mathrm C}$ is a central element of the decorated $p$-rack algebra $\hat {\cal Q}$.
\end{lemma}
\begin{proof} The proof is straightforward by means of the definition of 
the algebra $\hat {\cal Q}$ and the fact that $a \triangleright_{z_{ij}}$ and $\sigma_a^{z_{ij}}$ are bijective.
\end{proof}
From now on we consider, without loss of generality,  $\sum_{a\in X}h_a=1_{\hat  {\cal Q}}$ (see also Lemma \ref{lemmac}).

\begin{pro} \label{basic02}  (Algebra homomorphism)
Let $\hat {\cal Q}$ 
be the decorated $p$-rack algebra and 
${\cal R}^{z_{ij}}= \sum_{a \in X} h_a \otimes q^{z_{ij}}_a \in {\cal Q} \otimes {\cal Q}$ 
be a solution of the Yang-Baxter equation.
We also define for $z_{i,j,k} \in Y,$ $\Delta'_{z_{ij}}: \hat {\cal Q} \to \hat {\cal Q} \otimes \hat {\cal Q},$ such that for all $a \in X,$ 
\[\Delta'_{z_{jk}}((y^{z_{ik}}_a)^{\pm 1}) := (y^{z_{ij}}_a)^{\pm 1} \otimes (y^{z_{ik}}_a)^{\pm 1}, \quad \Delta'_{z_{ij}}(h_a) := \sum_{b,c \in X} 
h_b \otimes h_c\Big |_{b\triangleright_{z_{ij}} c = a}, \quad y_a^{z_{ik}} \in \{q_a^{z_{ik}},\ w_a^{z_{ik}}\}.\]
Then the following statements hold: 
\begin{enumerate}
\item $\Delta'_{z_{ij}}$ is a $\hat {\cal Q}$ algebra homomorphism for all $z_{i,j} \in Y.$
\item ${\cal R}^{z_{jk}} \Delta'_{z_{jk}}(y_a^{z_{ik}}) = \Delta^{'(op)}_{z_{kj}}(y_a^{z_{ik}}) {\cal R}^{z_{jk}},$ 
for $y_a^{z_{ik}} \in \{q_a^{z_{ik}},\ w_a^{z_{ik}}\},$ $a\in X,$ $z_{i,j,k} \in Y.$ Recall, $\Delta^{'(op)}_{z_{ij}} := \pi \circ \Delta'_{z_{ij}},$ where $\pi$ is the flip map.
\end{enumerate}
\end{pro}
\begin{proof} Bearing in mind Proposition \ref{basic01} it suffices to show:\\ $(1)$ $ $ 
$(i)$ $~\Delta'_{z_{ij}}(w_a^{z_{kj}}) \Delta'_{z_{ij}}(w_b^{z_{lj}}) = \Delta'_{z_{ij}}(w_{\sigma_a^{z_{lk}}(b)}^{z_{lj}}) \Delta'_{z_{ij}}(w_{\tau_b^{z_{lk}}(a)}^{z_{kj}}),$\\ $~~~~(ii)$
$~\Delta'_{z_{ij}}(w_a^{z_{kj}}) \Delta'_{z_{ij}}(q_a^{z_{lj}}) = \Delta'_{z_{ij}}(q_{\sigma_a^{z_{lk}}(b)}^{z_{lj}}) \Delta'_{z_{ij}}(w_a^{z_{kj}}),$\\
$~~~~(iii)$ $~\Delta'_{z_{ij}}(w_a^{z_{kj}}) \Delta'_{z_{ij}}(h_b) = \Delta'_{z_{ij}}(h_{\sigma^{z_{jk}}_a(b)}) \Delta'_{z_{ij}}(w_a^{z_{kj}})$ and
$(2)$ $~{\cal R}^{z_{jk}} \Delta'_{z_{jk}}(w_a^{z_{ik}}) = \Delta^{'(op)}_{z_{kj}}(w_a^{z_{ik}}) {\cal R}^{z_{jk}}.$

\begin{enumerate}
\item First it is straightforward, via the algebraic relations of the decorated $p$-rack  
algebra and the form of the coproducts to show $(i),$ $(ii)$. To show $(iii)$ we use in addition the condition, $\sigma^{z_{ik}}_c(b) \triangleright_{z_{ij}} \sigma^{z_{jk}}_{c}(a) = \sigma^{z_{jk}}_c(b \triangleright_{z_{ij}} a)$ (see Proposition \ref{qua2}).

\item This part is also shown directly by means of the relations of the $p$-decorated algebra.
\hfill \qedhere
\end{enumerate}
\end{proof}

\begin{defn}  
\label{setalgd2} ($p$-set Yang-Baxter algebra.) 
Let ${\cal Q}$ 
be a restricted $p$-rack algebra. Let also $\sigma^{z_{ij}}_a, \ \tau^{z_{ij}}_b: X\to X,$ and $\sigma^{z_{ij}}_a$ 
be bijective for all $a\in X,$ $z_{i,j} \in Y.$  We say that the unital, associative algebra $\hat {\cal Q}$ over $k,$
generated by indeterminates $1_{\hat {\cal Q}}$ (unit element), $q^{z_{ij}}_a, (q^{z_{ij}}_a)^{-1}, h_a, \in {\cal Q},$ 
$w^{z_{ij}}_a, (w^{z_{ij}}_a)^{-1}\in \hat {\cal Q},$ for $a \in X,$ $z_{i,j} \in Y$
and relations, (\ref{qualgbb}) is a $p$-set Yang-Baxter algebra.
\end{defn}

\begin{pro} \label{basica2b} 
Let $\hat {\cal Q}$ 
be a $p$-set Yang-Baxter algebra and
 ${\cal R}^{z_{ij}} = \sum_{b \in X} h_b\otimes q^{z_{ij}}_b \in \hat {\cal Q} \otimes \hat {\cal Q}$ 
is a solution of the Yang-Baxter equation. Let also for all $a,b,c \in X,$  $z_{i,j,k} \in Y:$  
$(b\triangleright_{z_{ij}} a) \bullet_{z_{jk}}(b\triangleright_{z_{ik}} c) = b \triangleright_{z_{ik}}(a \bullet_{z_{jk}}c)$ and
\begin{equation}
 \sigma^{z_{ik}}_c(a) \bullet_{z_{ij}} \sigma^{z_{jk}}_c(b) = \sigma_c^{z_{jk}}(a\bullet_{z_{ij}} b)
 \label{condition0}
\end{equation} 
We also define for $z_{i,j,k} \in Y,$ $\Delta_{z_{ij}}: \hat {\cal Q} \to \hat {\cal Q} \otimes 
\hat {\cal Q},$ such that for all $a \in X:$
\[\Delta_{z_{ij}}(q^{z_{kj}}_a) := q^{z_{ki}}_a \otimes q^{z_{kj}}_a, \quad \Delta_{z_{ij}}(h_a) := \sum_{b,c \in X} 
h_b \otimes h_c\Big |_{b\bullet_{z_{ij}} c = a}, \quad \Delta_{z_{ij}}(w_a^{z_{kj}}) = w_{a}^{z_{ki}} \otimes w_a^{z_{kj}}.\]
Then the following statements hold: 
\begin{enumerate}
\item $\Delta_{z_{ij}}$ is an algebra homomorhism for all ${z_{i,j}} \in Y.$
\item  $\Delta^{(op)}_{z_{ji}}(y) {\cal R}^{z_{ij}} = {\cal R}^{z_{ij}} \Delta_{z_{ij}}(y),$ 
$y \in \{h_a,\ q_a^{kj},\ w_a^{kj}\},$ for all $a \in X,$ $z_{i,j,k} \in Y,$
recall also that $\Delta^{(op)}_{z_{ij}} := \pi \circ \Delta_{z_{ij}},$ where $\pi$ is the flip map.
\end{enumerate}
\end{pro}
\begin{proof}  
In our proof below we are using the definition of the $p$-set Yang-Baxter algebra and (\ref{condition0}).
\begin{enumerate}
\item Recalling part 1 of Theorem \ref{basica1} it is sufficient to check the 
consistency of the following algebraic relations using (\ref{condition0}), for all $a, b\in X,$ $z_{i,j,k,l} \in Y:$
\begin{eqnarray}
&& \Delta_{z_{ij}}(w^{z_{lj}}_a) \Delta_{z_{ij}}(w^{z_{kj}}_b) = \Delta_{z_{ij}}(w^{z_{kj}}_{\sigma^{z_{kl}}_a(b)}) \Delta_{z_{ij}}(w^{z_{lj}}_{\tau^{z_{kl}}_b(a)}),  \nonumber\\   
&&\Delta_{z_{ij}}(w^{z_{kj}}_a) \Delta_{z_{ij}}(h_b) = \Delta_{z_{ij}}(h_{\sigma^{z_{jk}}_a(b)}) \Delta_{z_{ij}}(w^{z_{kj}}_a). \nonumber\\
&&\Delta_{z_{ij}}(w^{z_{lj}}_a) \Delta_{z_{ij}}(q^{z_{kj}}_b) = \Delta_{z_{ij}}(q^{z_{kj}}_{\sigma^{z_{kl}}_a(b)}) \Delta_{z_{ij}}(w^{z_{lj}}_{a}).\nonumber 
\end{eqnarray}
By using the algebra $\hat {\cal Q}$ relations and the form of the coproducts we show by direct computation 
that the above equations hold. In order to show the second of the three equalities above we also use identity (\ref{condition0}).

\item  Given Theorem \ref{basica1} it suffices to show that for all $a \in X,$ $z_{i,j,k} \in Y,$
$~\Delta_{z_{ij}}(w^{z_{kj}}_a) {\cal R}^{z_{ij} }= {\cal R}^{z_{ij}} \Delta_{z_{ij}}(w^{z_{kj}}_a).$ Indeed, this is shown by a direct computation using the algebraic relations of 
Definition \ref{setalgd}.
\hfill \qedhere
\end{enumerate}
\end{proof}

\begin{lemma} \label{basicc2} (Parametric (co)-associativity.) 
Let $\hat {\cal Q}$ 
be the decorated $p$-rack algebra and consider the subalgebra $\hat {\cal Q}^-$ consisting of 
the elements $1_{\hat {\cal Q}}, \ (q^{z_{ij}}_a)^{\pm 1},\ (w^{z_{ij}}_a)^{\pm 1}.$ We also  
define for all $z_{i,1,2, \ldots, n} \in Y,$ $a \in X,$ and $y^{z_{ij}}_a =\{q^{z_{ij}}_a, 
w^{z_{ij}}_a\},$ the map
$\Delta^{(n)}_{z_{1 2 \ldots n}}: \hat {\cal Q}  
\to \hat {\cal Q}^{\otimes n},$ such that
\begin{eqnarray}
    && \Delta^{(n)}_{z_{12\ldots n}}((y^{z_{in}}_a)^{\pm 1}) := (y^{z_{i1}}_a)^{\pm 1} 
    \otimes (y^{z_{i2}}_a)^{\pm 1}\otimes \ldots \otimes (y_a^{z_{in}})^{\pm 1}.
\end{eqnarray}

Then, for all $z_{i, 1,2, \ldots n} \in Y,$ 
\begin{enumerate}

\item  $\Delta_{z_{12...n}}^{(n)}((y_a^{z_{in}})^{\pm 1}):= (\Delta^{(n-1)}_{z_{12...n-1}} 
\otimes \id)\Delta_{z_{n-1n}}((y_a^{z_{in}})^{\pm 1})  = (\id 
\otimes \Delta^{(n-1)}_{z_{23...n}}) \Delta_{z_{1n}}((y_a^{z_{in}})^{\pm 1}),$ $a \in X.$
\item $\Delta^{(n)}_{z_{12\ldots n}}$ is a $\hat {\cal Q}^-$ algebra homomorphism.
\end{enumerate}
\end{lemma}
\begin{proof} 
The proof of both (1) and (2) follows directly from the form of the $n$-coproduct and the algebraic relations \ref{qualgd}.
\end{proof}

We come now to the derivation of the admissible Drinfel'd twist that will be used to provide 
the universal ${\cal R}$-matrix associated to the general set-theoretic solution. The $n$-fold twist will be also explicitly derived.
It is useful to introduce some practical notation that can be applied in the following propositions:  
let $i,j,k \in \{1,2,3\},$ then ${\cal F}^{z_{ijk}}_{jik} = \pi_{ij} \circ {\cal F}^{z_{ijk}}_{ijk}$ and 
${\cal F}^{z_{ijk}}_{ikj} = \pi_{jk} \circ {\cal F}^{z_{ijk}}_{ijk},$ where  $\pi$ is the flip map. 

\begin{thm} \label{twist2} (Drinfel'd twist \cite{Drinfel'd}) 
Let ${\cal R}^{z_{ij}} = \sum_{a\in X} h_a \otimes q^{z_{ij}}_a\in {\cal Q} \otimes {\cal Q}$ 
be the $p$-rack universal ${\cal R}-$matrix.  Let also $\hat {\cal Q}$ be the decorated $p$-rack algebra and
${\cal F}^{z_{ij}}\in \hat {\cal Q} \otimes \hat  {\cal Q},$
such that ${\cal F}^{z_{ij}} =\sum_{b\in X}h_b \otimes (w^{z_{ij}}_b)^{-1},$ for all $z_{i,j} \in Y$ and 
${\cal F}^{z_{ji}}_{ji} {\cal R}^{z_{ij}}_{ij} ={\cal R}^{Fz_{ij}}_{ij} {\cal F}^{z_{ij}}_{ij}.$
We also define:
\begin{eqnarray}
{\cal F}^{z_{123}}_{1,23} := \sum_{a\in X} h_a\otimes (w^{z_{12}}_a)^{-1} \otimes 
(w^{z_{13}}_a)^{-1},  \quad {\cal F}^{*z_{123}}_{12,3} : = \sum_{a,b \in X}h_a\otimes h_{\sigma^{z_{21}}_a(b)} \otimes (w^{z_{23}}_{b})^{-1} (w^{z_{13}}_a)^{-1}. \label{twi1}
\end{eqnarray}
Let also  for every
$a,b, \in X,$ $z_{i,j} \in Y,$ $~b\triangleright_{z_{ij}} a = 
\sigma^{z_{ji}}_b(\tau^{z_{ij}}_{(\sigma^{z_{ij}}_a)^{-1}(b)}(a)).$ 
Then, the following statements are true for $z_{1,2,3} \in Y:$
\begin{enumerate}
\item ${\cal F}^{z_{12}}_{12} {\cal F}^{*z_{123}}_{12,3} ={\cal F}^{z_{23}}_{23} {\cal F}^{z_{123}}_{1,23} =: 
{\cal F}^{z_{123}}_{123}.$
\item  (i) ${\cal F}^{z_{ikj}}_{ikj} {\cal R}^{z_{jk}}_{jk} ={\cal R}^{Fz_{jk}}_{jk} {\cal F}^{z_{ijk}}_{ijk}.$
\\
(ii)  ${\cal F}^{z_{jik}}_{jik} {\cal R}^{z_{ij}}_{ij} ={\cal R}^{Fz_{ij}}_{ij} {\cal F}^{z_{ijk}}_{ijk}.$
\end{enumerate}
 The element ${\cal F}^{z_{ij}}$ is called an admissible Drinfel'd twist.
\end{thm}
\begin{proof}
The proof is straightforward based on the underlying algebra $\hat {\cal Q}.$
\begin{enumerate}
\item Indeed, this is proved by a direct computation and use of the decorated $p$-rack algebra. 
In fact, ${\cal F}^{z_{123}}_{123} = \sum_{a,b\in X} h_a \otimes h_b (w_a^{z_{12}})^{-1} 
\otimes (w_b^{z_{23}})^{-1} (w_a^{z_{13}})^{-1}.$
\item Given the notation introduced before the theorem it  suffices to show  that
${\cal F}^{z_{132}}_{132}{\cal R}^{z_{23}}_{23} = {\cal R}^{Fz_{23}}_{23} {\cal F}^{z_{123}}_{123}$ and 
${\cal F}^{z_{213}}_{213}  {\cal R}^{z_{12}}_{12} = {\cal R}^{Fz_{12}}_{12} {\cal F}^{z_{123}}_{123}.$ \\
\noindent (i) 
Due to the fact that for all $a \in X,$ $z_{i,j} \in Y,$
$ \Delta_{z_{ij}}(w^{z_{ki}}_a){\cal R}^{z_{ij}} =  {\cal R}^{z_{ij}} \Delta_{z_{ij}}(w^{z_{ki}}_a)$ 
(see Proposition \ref{basica2b})
we arrive at ${\cal F}^{z_{132}}_{1,32} {\cal R}^{z_{23}}_{23} = {\cal R}^{z_{23}}_{23} F^{z_{123}}_{1,23},$ then
\[{\cal F}^{z_{132}}_{132} {\cal R}^{z_{23}}_{23} = {\cal F}^{z_{32}}_{32} {\cal F}^{z_{132}}_{1,32} {\cal R}^{z_{32}}_{23} ={\cal F}^{z_{32}}_{32} 
{\cal R}^{z_{32}}_{23} {\cal F}^{z_{123}}_{1,23} = {\cal R}_{23}^{Fz_{23}}{\cal F}^{z_{123}}_{123}. \]
\noindent (ii) By using the relations 
of the decorated  $p$-rack algebra $\hat {\cal Q}$ we compute:
\begin{eqnarray}
&& {\cal F}^{*z_{213}}_{21,3}{\cal R}^{z_{12}}_{12} = \sum_{a,c\in X} h_a \otimes q^{z_{12}}_a 
h_{a\triangleright_{z_{12}} c} \otimes 
 (w^{z_{23}}_c  w^{z_{13}}_{(\sigma^{z_{12}}_c)^{-1}(a)} )^{-1},\nonumber \\ 
&& {\cal R}^{z_{12}}_{12}{ \cal F}^{*z_{123}}_{12,3}= \sum_{a,b \in X} 
h_a \otimes q^{z_{12}}_a h _{\sigma^{z_{21}}_a(b)} \otimes (w^{z_{13}}_aw^{z_{12}}_b)^{-1}.\nonumber
\end{eqnarray}
Using also,  $~b\triangleright_{z_{12}} a = \sigma^{z_{21}}_b(\tau^{z_{12}}_{(\sigma^{z_{12}}_a)^{-1}(b)}(a))$ and $w^{z_{kj}}_a w^{z_{ij}}_b = w^{z_{ij}}_{\sigma^{z_{ik}}_a(b)} w^{z_{kj}}_{ \tau^{z_{ik}}_b(a)}$
we conclude that ${\cal F}^{*z_{213}}_{21,3}{\cal R}^{z_{12}}_{12} = 
{\cal R}_{12}{ \cal F}^{*z_{123}}_{12,3}$ and consequently (recall 
${\cal F}^{z_{213}}_{213} = {\cal F}^{z_{21}}_{21} {\cal F}^{*z_{213}}_{21,3}$)
\[{\cal F}^{z_{213}}_{213} {\cal R}^{z_{12}}_{12} = {\cal F}^{z_{21}}_{21} {\cal F}^{*z_{213}}_{21,3} {\cal R}^{z_{12}}_{12} 
={\cal F}^{z_{21}}_{21} {\cal R}^{z_{12}}_{12}{\cal F}^{*z_{123}}_{12,3} = {\cal R}_{12}^{Fz_{12}}{\cal F}^{z_{123}}_{123}. \hfill \qedhere
\]
\end{enumerate}
\end{proof}

\begin{cor} \cite{Drinfel'd, Drinfel'd2} Let $Y$ be a non-empty set and $z_{i,j}\in Y.$ 
Let also ${\cal F}^{z_{ij}}$ be an admissible twist and ${\cal R}^{z_{ij}}$ be a solution of the Yang-Baxter equation. 
Then ${\cal R}^{Fz_{ij}}:= ({\cal F}^{z_{ji}})^{(op)} {\cal R}^{z_{ij}} ({\cal F}^{z_{ij}})^{-1}$ 
($({\cal F}^{z_{ji}})^{(op)} = \pi \circ {\cal F}^{z_{ji}},$ $\pi$ is the flip map) 
is also a solution of the Yang-Baxter equation.
\end{cor}
\begin{proof}
The proof is quite straightforward,  \cite{Drinfel'd, Drinfel'd2} 
(see also proof in \cite{Doikou1} for set-theoretic solutions),  
we just give a brief outline here: if ${\cal F}^{z_{ij}}$ is admissible, 
then from the Yang-Baxter equation and due to Proposition \ref{twist2}, $z_{1,2,3}\in Y:$
\begin{equation}
{\cal  F}^{z_{321}}_{321} {\cal R}^{z_{12}}_{12} {\cal R}^{z_{13}}_{13} {\cal R}^{z_{23}}_{23} ={\cal  F}^{z_{321}}_{321} 
{\cal R}^{z_{23}}_{23}{\cal R}^{z_{13}}_{13}{\cal R}^{z_{12}}_{12}\ \Rightarrow\  {\cal R}^{Fz_{12}}_{12} 
{\cal R}^{Fz_{13}}_{13} {\cal R}^{Fz_{23}}_{23}{\cal  F}^{z_{123}}_{123} 
=  {\cal R}^{Fz_{23}}_{23}{\cal R}^{Fz_{13}}_{13}{\cal R}^{Fz_{12}}_{12}{\cal  F}^{z_{123}}_{123}. \nonumber
\end{equation}
But ${\cal F}^{z_{123}}_{123}$ is invertible, hence ${\cal R}^{Fz_{ij}}$ indeed satisfies the YBE.
\end{proof}

\begin{lemma} (The n-fold twist.)  
Let $\hat {\cal Q}$ be the decorated $p$-rack algebra. Let also ${\cal R}^{z_{ij}} = 
\sum_{a\in X }h_a \otimes q^{z_{ij}}_a \in \hat {\cal Q} \otimes \hat {\cal Q}$ 
be a solution of the Yang-Baxter equation and ${\cal F}^{z_{ij}} \in \hat {\cal Q} \otimes \hat {\cal Q},$ such that 
${\cal F}^{z_{ij}} = \sum_{a \in X} h_a\otimes (w_a^{z_{ij}})^{-1},$ $z_{i,j} \in X$.
Define also for all $z_{1,2, \ldots ,n} \in Y:$
\begin{eqnarray}
  {\cal F}^{z_{12\ldots n}}_{1, 2 3\ldots n} &:=& \sum_{a\in X} h_a \otimes \Delta_{z_{2\ldots n}}^{(n-1)}((w_a^{z_{12}})^{-1}) =\sum_{a\in X} h_a \otimes (w_a^{z_{12}})^{-1}  \otimes(w_a^{z_{13}})^{-1} \otimes   \ldots \otimes (w_a^{z_{1n}})^{-1}, \nonumber\\ 
    {\cal F}^{*z_{12\ldots n}}_{1 2 \ldots n-1, n} &:=& \sum_{a_1, a_2, \ldots, a_{n-1} \in X} 
   h_{a_1} \otimes h_{\sigma^{z_{21}}_{a_1}(a_2)} \otimes h_{\sigma^{z_{31}}_{a_1}(\sigma^{z_{32}}_{a_2}(a_3))} \otimes \ldots \nonumber \\ 
   & \otimes &
   h_{\sigma^{z_{n-11}}_{a_1}(\sigma^{z_{n-12}}_{a_2}(\ldots 
   \sigma^{z_{n-1 n-2}}_{a_{n-2}}(a_{n-1})) \ldots )} 
\otimes (w_{a_{n-1}}^{z_{n-1n}})^{-1}
(w_{a_{n-2}}^{z_{n-2n}})^{-1} \ldots (w_{a_{1}}^{z_{1n}})^{-1} \nonumber 
   \end{eqnarray}
Then,
\begin{enumerate}
\item ${\cal F}^{z_{2\ldots n}}_{2\ldots n}
{\cal F}^{z_{12\ldots n}}_{1,2\ldots n} =
{\cal F}^{z_{12\ldots n-1}}_{12\ldots n-1} {\cal F}^{*z_{12\ldots n}}_{12\ldots n-1,n}=: {\cal F}^{z_{1 2\ldots n}}_{12\ldots n}.$ 
\\
\item The explicit expression of the $n$-fold twist is given as
\begin{eqnarray}
{\cal F}^{z_{12\ldots n}}_{12\ldots n} 
&=&
\sum_{a_1, a_2, \ldots, a_{n-1} \in X} h_{a_1} \otimes h_{a_2} (w_{a_1}^{z_{12}})^{-1} \otimes 
h_{a_3} (w_{a_2}^{z_{23}})^{-1}(w_{a_1}^{z_{13}})^{-1} \otimes \ldots \otimes \nonumber \\
& & h_{a_{n-1}}(w_{a_{n-2}}^{z_{n-2 n-1}})^{-1} \ldots (w_{a_1}^{z_{1n-1}})^{-1} \otimes (w^{z_{n-1n}}_{a_{n-1}})^{-1}(w_{a_{n-2}}^{z_{n-2 n}})^{-1} \ldots (w_{a_1}^{z_{1n}})^{-1}. 
 \label{nfold}
\end{eqnarray}

\item ${\cal F}_{1,2 3 \ldots j+1 j\ldots n}^{z_{12 \ldots j+1 j \ldots n}} {\cal R}^{z_{jj+1}}_{j j+1} = {\cal R}^{z_{jj+1}}_{jj+1} {\cal F}_{1,2 3 \ldots j j+1\ldots n}^{z_{12 \ldots j j+1 \ldots n}},$ $~n-1 \geq j>1,$\\
\\
${\cal F}_{12 \ldots j+1 j\ldots n-1, n}^{z_{12 \ldots j+1 j \ldots n}} {\cal R}^{z_{jj+1}}_{j j+1} = {\cal R}^{z_{jj+1}}_{jj+1} {\cal F}_{1 2 \ldots j j+1\ldots n-1,  n}^{z_{12 \ldots j j+1 \ldots n}},$ $~n-1>j\geq 1,$
\\
\\
${\cal F}_{12 \ldots j+1 j\ldots n}^{z_{12 \ldots j+1 j \ldots n}} {\cal R}_{j j+1} = {\cal R}^{Fz_{jj+1}}_{jj+1} {\cal F}_{1 2 \ldots j j+1\ldots n}^{z_{12 \ldots j j+1 \ldots n}},$ $~n-1 \geq j\geq 1.$ \end{enumerate}
\end{lemma}
\begin{proof} 
These statements are proven by iteration and direct computation by using the $\hat {\cal Q}$ algebra relations. Part (2) of Theorem \ref{twist2} is also used in proving (3).
\end{proof}

\begin{cor} \label{corf}
    Let ${\cal F}^{z_{12\ldots n}}_{12\ldots n}$ be the $n$-fold twist \ref{nfold}. 
    Let also $(X, \circ)$ be a group, 
    $z_i \circ z_j = z_j \circ z_i$ and 
    $w_a^{z_{jk}}w_b^{z_{ik}} = w_{a\circ b}^{z_{i\circ j k}}$ 
    for all $a,b \in X,$ $z_{i,j,k} \in Y,$ 
    where $z_{i\circ j k}$ 
    denotes dependence on $(z_{i \circ j},  z_k)$. Then for $z_{1,2,\ldots,n} \in Y,$
    \[{\cal F}^{z_{12\ldots n}}_{12\ldots n} = 
    \sum_{a_1, \ldots, a_{n}\in X } h_{a_1} \otimes h_{a_2} (w^{z_{12}}_a)^{-1}\otimes \ldots \otimes h_{a_{n-1}} (w^{z_{1\circ 2 \ldots \circ n-2 n-1}}_{a_1\circ a_2 \circ \ldots \circ a_{n-1} })^{-1} \otimes (w^{z_{1\circ 2 \ldots \circ n-1 n}}_{a_1\circ a_2 \circ \ldots \circ a_{n} })^{-1}.\]
\end{cor}
\begin{proof}
    This is a consequence of the form of the 
    $n$-twist (\ref{nfold}) and relation $w_a^{z_{jk}}w_b^{z_{ik}} = w_{a\circ b}^{z_{i\circ j k}}$ for all $a,b \in X,$ $z_{i,j,k} \in Y.$ 
    \end{proof}

\begin{rem} \label{rtwi}
(Twisted universal ${\cal R}$-matrix) 
We recall the admissible twist ${\cal F}^{z_{12}} = \sum_{b\in X}h_b \otimes (w_b^{z_{12}})^{-1},$ and the universal $p$-rack ${\cal R}$-matrix is ${\cal R}^{z_{12}} = \sum_{a\in X} h_a\otimes q_a^{z_{12}}$, then we obtain:
\begin{itemize}

\item The twisted ${\cal R}-$matrix: 
\[{\cal R}^{Fz_{12}} = ({\cal F}^{z_{21}})^{(op)} {\cal R}^{z_{12}} ({\cal F}^{z_{12}})^{-1} \]

\item The twisted coproducts: for $z_{12} \in Y,$
$\Delta_{z_{12}}^F(y) = 
{\cal F}^{z_{12}} \Delta_{z_{12}}(y) 
({\cal F}^{z_{12}})^{-1},$ $y\in \hat {\cal Q}$  
and we recall ({in the case of the $p$-set theoretic Yang-Baxter algebra, Proposition (\ref{basica2b})}), for $a\in X,$ $z_{i,1,2} \in Y$
\[ \Delta_{z_{12}}(w^{z_{i2}}_a) =w^{z_{i1}}_a\otimes w^{z_{i2}}_a,\quad \Delta_{z_{12}}(h_a) = 
\sum_{b,c\in X} h_b \otimes h_c\Big |_{b\bullet_{z_{12}} c =a}, \quad \Delta_{z_{12}}(q^{z_{i2}}_a) =q^{z_{i1}}_a \otimes q^{z_{i2}}_a.\]
\end{itemize}
Moreover it  follows that ${\cal R}^{Fz_{21}} \Delta_{z_{12}}^F(y) = \Delta_{z_{12}}^{F(op)}(y) {\cal R}^{Fz_{12}},$ $y \in \hat {\cal Q},$ $z_{1,2} \in Y.$
\end{rem}

\begin{rem} \label{remfu2b}  
{\bf Fundamental representation $\&$ the set-theoretic solution:}\\  
Let $\hat {\cal Q}$ be the decorated $p$-rack algebra and $\rho: \hat  {\cal Q} \to {\EEnd({\mathbb C}^n)},$ such that
\begin{equation}
q^{z_{ij}}_a \mapsto \sum_{x \in X} e_{x, a \triangleright_{z_{ij}} x}, \quad h_a\mapsto e_{a,a}, \quad  
w^{z_{ij}}_a \mapsto \sum_{b \in X} e_{\sigma^{z_{ji}}_a(b),b}.\label{repbb}
\end{equation}
This is a representation of the decorated $p$-rack algebra if and only if $(X, \triangleright_{z_{ij}})$ is a rack for all $z_{i,j} \in Y$ and $\sigma^{z_{ik}}_{a}(\sigma^{z_{ij}}_b(c)) = \sigma^{z_{ij}}_{\sigma^{z_{jk}}_{a}(b)}(\sigma^{z_{ik}}_{\tau^{z_{jk}}_b(a)}(c))$ for all $a,b,c \in X.$
Then $ {\cal R}^{Fz_{ij}} \mapsto R^{Fz_{ij}}= 
\sum_{a,b\in X} e_{b,\sigma^{z_{ij}}_a(b)} \otimes e_{a, \tau^{z_{ij}}_b(a)},$
where we recall that $\tau^{z_{ij}}_b(a):=
\sigma^{z_{ji}}_{(\sigma^{z_{ij}}_a)^{-1}(b)}
(\sigma^{z_{ij}}_a(b) 
\triangleright_{z_{ij}} a).$ We note that $R^{Fz_{ij}}$ is the linearized version of the set-theoretic solution, {and we also notice that $(R^{Fz_{12}})^{-1} = (R^{Fz_{12}})^T,$ where $^T$ denotes total transposition.}

The $n$-fold twist \ref{nfold} in the fundamental representation becomes:
\begin{eqnarray}
&& {\cal F}^{z_{12\ldots n}}_{12 \ldots n} \mapsto F^{z_{12\ldots n}}_{12 \ldots n} =\sum_{a_1, \ldots a_n \in X} e_{a_1, a_1} \otimes e_{a_2, \sigma_{a_1}^{z_{21}}(a_2)}  \otimes \ldots \otimes \nonumber \\ 
&& e_{a_{n-1}, \sigma^{z_{n-11}}_{a_1}
 (\sigma^{z_{n-12}}_{a_2}( \ldots \sigma^{z_{n-1n-2}}_{a_{n-2}}(a_{n-1}) \ldots ))} \nonumber \otimes e_{a_n, \sigma^{z_{n1}}_{a_1}(\sigma^{z_{n2}}_{a_2}( \ldots \sigma^{z_{nn-1}}_{a_{n-1}}(a_n) \ldots ))}. \\ & &\label{repf} 
 \end{eqnarray}
In the absence of parameters (non-parametric case) expression (\ref{repf}) reduces to the expression derived in \cite{Doikou1}
and should also coincide with the linearized version of the non-local transformation introduced in \cite{Sol}.
\end{rem}

Let the universal twisted ${\cal R}$-matrix of Remark \ref{rtwi} expressed in the compact form 
${\cal R}^{Fz_{ij}} = \sum_{a,b \in X} {\mathrm f}^{z_{ij}}_{b,a} \otimes {\mathrm g}^{z_{ij}}_{a,b},$ then from the Yang-Baxter equation, and after recalling the representations 
(\ref{repbb}):
$(\rho \otimes \id){\cal R}^{Fz_{ij}} := L^{Fz_{ij}} = \sum_{a\in X}e_{b, \sigma^{z_{ij}}_a(b)} \otimes {\mathrm g}^{z_{ij}}_{a,b},$ 
$~(\id \otimes \rho){\cal R}^{Fz_{ij}}:= \hat L^{Fz_{ij}} = \sum_{a\in X} {\mathrm f}^{z_{ij}}_{b,a} \otimes e_{a, \tau^{z_{ij}}_b(a)},$ and 
$(\rho \otimes \rho){\cal R}^{Fz_{ij}} :=R^{Fz_{ij}} = \sum_{a,b \in X} e_{b,\sigma^{z_{ij}}_a(b)} \otimes e_{a, \tau^{z_{ij}}_b(a)},$ 
the consistent algebraic relations (\ref{a}) are satisfied (the interested reader is also 
referred to \cite{DoiRyb22} for detailed computations). 
These  lead to the $p$-set Yang-Baxter algebra and provide a consistency check on the associated algebraic relations.

\begin{exa}
Consider the fundamental representation of the $p$-set Yang-Baxter algebra $\hat {\cal Q}$ (recall also Examples \ref{exa11} and \ref{exabasic}, e.g. for $\xi =1$), 
$\rho: \hat {\cal Q} \to \EEnd({\mathbb C}^n),$ $w^{z_{ij}}_a 
\mapsto \omega^{z_{ij}}_a:= 
\sum_{b \in X} e_{\sigma^{z_{ji}}_a(b) ,b}$ $a \in X$ $z_{i,j} \in Y,$ then:
\begin{eqnarray}
 \omega^{z_{jk}}_a \omega^{z_{ik}}_b =  \sum_{c\in X} e_{\sigma^{z_{kj}}_a(\sigma^{z_{ki}}_b(c)),c}. \label{conj}
\end{eqnarray}
Recall also the map for $z_{i,j} \in Y,$ $\sigma_a^{z_{ij}}: X \to X,$ $\sigma^{z_{ij}}_a(b) = z_i^{-1} -a\circ z_i^{-1} \circ z_j +a\circ b\circ z_j$ and
\begin{equation}
a\circ b = \sigma^{z_{ij}}_a(b) \circ \tau_b^{z_{ij}}(a) \quad \mbox{and} \quad  \sigma^{z_{kj}}_a(\sigma^{z_{ki}}_b(c)) = \sigma^{z_{k i\circ j}}_{a\circ b}(c),
\end{equation}
then via (\ref{conj}), we conclude that $\omega^{z_{jk}}_a \omega^{z_{ik}}_b = \omega^{z_{i\circ jk}}_{a\circ b} : = \omega^{z_{ik}}_{\sigma^{z_{ij}}_a(b)} \omega^{z_{jk}}_{\tau^{z_{ij}}_b(a)},$ where recall  the  shorthand notation $z_{i\circ j k}$ denotes dependence on $(z_i \circ z_j, z_k).$  

The $n$-fold twist in this case becomes (see also related Corollary \ref{corf}):
\begin{eqnarray}
 F^{z_{12\ldots n}}_{12 \ldots n} &=&\sum_{a_1, \ldots a_n \in X} e_{a_1, a_1} \otimes e_{a_2, \sigma_{a_1}^{z_{21}}(a_2)}  \otimes \ldots   \otimes  e_{a_{n-1}, \sigma^{z_{n-1 1\circ 2\circ \ldots \circ n-2}}_{a_1 \circ a_2\circ \ldots \circ a_{n-2}}
 (a_{n-1})} \nonumber \otimes e_{a_n, \sigma^{z_{n 1\circ 2\circ \ldots \circ n-1}}_{a_1 \circ a_2\circ \ldots \circ a_{n-1}}
 (a_{n})}.  \nonumber \label{repf2}
 \end{eqnarray}

\end{exa}

\subsection*{\texorpdfstring{The special $p$-set algebra}{}}
Recall the decorated $p$-rack algebra and consider the special case, where for all $a,b \in X,$ $z_{i,j}\in Y$ $a\triangleright_{z_{ij}} b = b,$ and consequently
$\sigma^{z_{ji}}_{\sigma^{z_{ij}}_a(b)}(\tau_b^{z_{ij}}(a)) = a.$ Let also, $q_a^{z_{ij}} = 1_{\hat{\cal Q}}$, then the decorated $p$-rack algebra reduces to the {\it special $p$-set algebra}.
\begin{defn}  
\label{setalgd22}  Let 
$\sigma^{z_{ij}}_a, \ \tau_b^{z_{ij}}: X\to X,$ and $\sigma^{z_{ij}}_a$ be a bijection for all $a\in X,$ $z_{i,j} \in Y$.
We say that the unital, associative algebra $\hat {\cal Q}$ over $k,$
generated by intederminates 
$h_a, w^{z_{ij}}_a, (w^{z_{ij}}_a)^{-1} \in \hat {\cal Q},$ $a \in X,$ 
$1_{\hat {\cal Q}}$ (the unit element)
and relations, for $a,b \in X,$ $z_{i,j,k} \in Y:$
\begin{eqnarray}
h_a  h_b =\delta_{a, b} h_a, 
\quad w^{z_{ij}}_a (w^{z_{ij}}_a)^{-1} =1_{\hat {\cal Q}}, ~~ 
w_a^{z_{ki}} w_b^{z_{ji}}= 
w^{z_{ji}}_{\sigma^{z_{jk}}_a(b)} w_{\tau^{z_{kj}}_{b}(a)}^{z_{ki}} ~~  w_a^{z_{ji}} h_b = h_{\sigma^{z_{ij}}_a(b)} w_a^{z_{ji}},   \label{qualgbb2}
\end{eqnarray}
is a {\it special $p$-set algebra}.
\end{defn}

In this case the $p$-rack universal ${\cal R}$-matrix reduces to the identity map, ${\cal R}^{z_{ij}} = \id,$ and the twisted ${\cal R}$-matrix is
\[{\cal R}^{Fz_{12}} = ({\cal F}^{z_{21}})^{(op)} ({\cal F}^{z_{12}})^{-1} = \sum_{a,b \in X} (w_a^{z_{ji}})^{-1} h_b \otimes h_a w_b^{z_{ij}}. \]
This is an reversible solution as it satisfies, ${\cal R}_{12}^{Fz_{12}}{\cal R}_{21}^{Fz_{21}} = \id.$ In the fundamental representation (Remark \ref{remfu2b}), ${\cal R}^{Fz_{ij}} \mapsto  R^{Fz_{ij}} = \sum_{a,b \in X} e_{b, \sigma^{z_{ij}}_{a}(b)} \otimes e_{a, \tau^{z_{ij}}_b(a)},$ i.e. it reduces to the linearized version of the reversible set-theoretic solution.

\subsection*{Acknowledgments}
\noindent 
Support from the EPSRC research grant  EP/V008129/1 is  acknowledged.

\end{document}